\pdfoutput=1

\RequirePackage{luatex85}

\documentclass[a4paper,11pt]{article}
\usepackage{iftex}

\ifluatex
\else
\ifxetex
\else
\usepackage[utf8]{inputenc}
\usepackage[T1]{fontenc}
\fi
\fi

\usepackage{amsmath}
\usepackage{amssymb}
\usepackage{amsthm}
\usepackage{mathtools}
\usepackage{thmtools}
\usepackage{thm-restate}

\ifluatex
\usepackage{fontspec}
\usepackage{unicode-math}
\else
\ifxetex
\usepackage{fontspec}
\usepackage{unicode-math}
\else
\usepackage{isomath}
\usepackage{libertine}
\usepackage[libertine]{newtxmath} % must load after ams packages
\usepackage[scaled=0.96]{zi4} % use Inconsolata as the typewriter font
\fi
\fi

\usepackage[DIV=11]{typearea} % large value = less margin, recommended values: 10pt: 8, 11pt: 10, 12pt: 12

\usepackage{microtype}

\usepackage[font=small]{caption}
\usepackage[algo2e,ruled,vlined,linesnumbered]{algorithm2e}
\usepackage{algorithm}
\usepackage{enumitem}
\usepackage{xcolor}
\usepackage{xspace} % automatic space after no-arguments macros if using \xspace

\usepackage[
	style=alphabetic,
    maxalphanames=4,
    minalphanames=4,
    maxnames=20,
    minnames=20,
	backref=true,
	doi=true,
	url=true,
	backend=biber
]{biblatex}

\ifpdf
\usepackage[ocgcolorlinks]{hyperref} % Option ocgcolorlinks makes links only colored on screen and not on printouts
\else
\usepackage[hidelinks]{hyperref}
\fi

\newcommand*\samethanks[1][\value{footnote}]{\footnotemark[#1]}

\allowdisplaybreaks[1]

\makeatletter
\g@addto@macro\bfseries{\boldmath}
\makeatother

\makeatletter
\g@addto@macro\mdseries{\unboldmath}
\g@addto@macro\normalfont{\unboldmath}
\g@addto@macro\rmfamily{\unboldmath}
\g@addto@macro\upshape{\unboldmath}
\makeatother

\makeatletter
\g@addto@macro\bfseries{\boldmath}
\def\thmhead@plain#1#2#3{%
  \thmname{#1}\thmnumber{\@ifnotempty{#1}{ }\@upn{#2}}%
  \thmnote{ {\the\thm@notefont\unboldmath(#3)}}}
\let\thmhead\thmhead@plain
\makeatother

\DeclareCiteCommand{\citem}
    {}
    {\mkbibbrackets{\bibhyperref{\usebibmacro{postnote}}}}
    {\multicitedelim}
    {}

\renewcommand*{\multicitedelim}{\addcomma\space}

\AtEveryCitekey{\clearfield{note}}

\makeatletter
\AtBeginDocument{%
    \newlength{\temp@x}%
    \newlength{\temp@y}%
    \newlength{\temp@w}%
    \newlength{\temp@h}%
    \def\my@coords#1#2#3#4{%
      \setlength{\temp@x}{#1}%
      \setlength{\temp@y}{#2}%
      \setlength{\temp@w}{#3}%
      \setlength{\temp@h}{#4}%
      \adjustlengths{}%
      \my@pdfliteral{\strip@pt\temp@x\space\strip@pt\temp@y\space\strip@pt\temp@w\space\strip@pt\temp@h\space re}}%
    \ifpdf
      \typeout{In PDF mode}%
      \def\my@pdfliteral#1{\pdfliteral page{#1}}% I don't know why % this command...
      \def\adjustlengths{}%
    \fi
    \ifxetex
      \def\my@pdfliteral #1{}% isn't equivalent to this one
      \def\adjustlengths{\setlength{\temp@h}{-\temp@h}\addtolength{\temp@y}{1in}\addtolength{\temp@x}{-1in}}%
    \fi%
    \def\Hy@colorlink#1{%
      \begingroup
        \ifHy@ocgcolorlinks
          \def\Hy@ocgcolor{#1}%
          \my@pdfliteral{q}%
          \my@pdfliteral{7 Tr}% Set text mode to clipping-only
        \else
          \HyColor@UseColor#1%
        \fi
    }%
    \def\Hy@endcolorlink{%
      \ifHy@ocgcolorlinks%
        \my@pdfliteral{/OC/OCPrint BDC}%
        \my@coords{0pt}{0pt}{\pdfpagewidth}{\pdfpageheight}%
        \my@pdfliteral{F}% Fill clipping path (the url's text) with
        \my@pdfliteral{EMC/OC/OCView BDC}%
        \begingroup%
          \expandafter\HyColor@UseColor\Hy@ocgcolor%
          \my@coords{0pt}{0pt}{\pdfpagewidth}{\pdfpageheight}%
          \my@pdfliteral{F}% Fill clipping path (the url's text)
        \endgroup%
        \my@pdfliteral{EMC}%
        \my@pdfliteral{0 Tr}% Reset text to normal mode
        \my@pdfliteral{Q}%
      \fi
      \endgroup
    }%
}
\makeatother

\usepackage[disable]{todonotes}

\newcommand{\antonis}[1]{\todo[linecolor=orange!50!black,backgroundcolor=orange!25,bordercolor=orange!50!black]{\scriptsize \textbf{AS:} #1}}

\colorlet{DarkRed}{red!50!black}
\colorlet{DarkGreen}{green!50!black}
\colorlet{DarkBlue}{blue!50!black}

\hypersetup{
	linkcolor = DarkRed,
	citecolor = DarkGreen,
	urlcolor = DarkBlue,
	bookmarksnumbered = true,
	linktocpage = true
}

\DontPrintSemicolon
\SetKwProg{Procedure}{Procedure}{}{}
\SetProcNameSty{textsc}
\SetFuncSty{textsc}
\SetKw{KwAnd}{and}
\SetKw{KwDo}{do}

\SetCommentSty{mycommfont}

\declaretheorem[numberwithin=section]{theorem}
\declaretheorem[numberlike=theorem]{lemma}

\declaretheorem[numberlike=theorem]{Definition}

\declaretheorem[numberlike=theorem]{observation}

\newcommand{\tbrs}{$(2, \beta)$-ruling set\xspace}
\newcommand{\abrs}{$(\alpha, \beta)$-ruling set\xspace}
\newcommand{\kbabrs}{$k$-bounded $(\alpha, \beta)$-ruling set\xspace}
\newcommand{\kbtbrs}{$k$-bounded $(2, \beta)$-ruling set\xspace}
\newcommand{\ttrs}{$(2, 2)$-ruling set\xspace}
\newcommand{\kbttrs}{$k$-bounded $(2, 2)$-ruling set\xspace}

\newcommand{\tors}{$(2, 1)$-ruling set\xspace}
\newcommand{\kbtors}{$k$-bounded $(2, 1)$-ruling set\xspace}

\DeclareMathOperator*{\argmax}{argmax}

\title{Dynamic algorithms for $k$-center on graphs\thanks{Supported by the Austrian Science Fund (FWF): P 32863-N. This project has received funding from the European Research Council (ERC) under the European Union's Horizon 2020 research and innovation programme (grant agreement No 947702).}}

\author{
  Emilio Cruciani\thanks{Department of Computer Science, University of Salzburg, Austria.}
  \and 
  Sebastian Forster\samethanks[1]
  \and
  Gramoz Goranci\thanks{Faculty of Computer Science, University of Vienna, Vienna, Austria.}
  \and
  Yasamin Nazari\thanks{Department of Computer Science, VU Amsterdam. This research was partially conducted when the author was a postdoc at University of Salzburg.}
  \and 
  Antonis Skarlatos\samethanks[1]
}

\date{}

\ifpdf
\hypersetup{
	pdftitle = {Dynamic algorithms for k-center on graphs},
	pdfauthor = {Emilio Cruciani, Sebastian Forster, Gramoz Goranci, Yasamin Nazari, Antonis Skarlatos}
}
\else
\fi

\begin{document}
\maketitle
\begin{abstract}
In this paper we give the first efficient algorithms for the $k$-center problem on dynamic graphs undergoing edge updates. In this problem, the goal is to partition the input into $k$ sets by choosing $k$ centers such that the maximum distance from any data point to its closest center is minimized. It is known that it is NP-hard to get a better than $2$ approximation for this problem.

While in many applications the input may naturally be modeled as a graph, all prior works on $k$-center problem in dynamic settings are on point sets in arbitrary metric spaces. 
In this paper, we give a deterministic decremental $(2+\epsilon)$-approximation algorithm and a randomized incremental $(4+\epsilon)$-approximation
algorithm, both with amortized update time $kn^{o(1)}$ for weighted graphs.  
Moreover, we show a reduction that leads to a fully dynamic $(2+\epsilon)$-approximation algorithm
for the $k$-center problem, with worst-case update time that
is within a factor $k$ of the state-of-the-art fully dynamic $(1+\epsilon)$-approximation single-source shortest
paths algorithm in graphs.
Matching this bound is a natural goalpost because the approximate distances of each vertex to its center can be used to maintain a $(2+\epsilon)$-approximation of the graph diameter and the fastest known algorithms for such a diameter approximation also rely on maintaining approximate single-source distances.
\end{abstract}

% \newpage
% \tableofcontents
% \newpage

\section{Introduction} \label{sec:intro}
Clustering is a key concept in data analysis that involves organizing `similar' data into groups. One of the most fundamental and well-studied objectives is the \emph{$k$-center objective}. Specifically, given a metric space with $n$ points and a positive integer $k \leq n$, the goal of the \emph{$k$-center problem} is to select $k$ points, referred to as \emph{centers}, such that the maximum distance of any point in the metric space to its closest center is minimized. 
It is known that $k$-center is NP-hard to approximate within a factor of $(2-\epsilon)$ for any $\epsilon > 0$~\cite{HsuN79}.
Due to its popularity, $k$-center has been considered under several algorithmic frameworks, including approximation algorithms~\cite{HsuN79,Gonzalez85,HochbaumS86,plesnik1980computational,FederG88}, parameterized complexity~\cite{Feldmann15,BandyapadhyayFM22}, massive parallel computation (MPC) model~\cite{CeccarelloPP19,BeraDGK22,BergBM23}, and beyond worst-case analysis~\cite{BalcanHW20}, among others. This problem also serves as a testbed for developing fundamental algorithmic definitions and paradigms, which are then often applied to solving other variants of clustering objectives.

Clustering in the \emph{dynamic setting} has received increasing attention in recent years. 
This line of work was initiated by Charikar et al.~\cite{DBLP:conf/stoc/CharikarCFM97} who considered the problem of minimizing cluster diameters under the insertions of new points in an underlying metric space.
Under both point insertions and deletions, the $k$-center problem was considered by Chan et al.~\cite{chan2018fully}, who achieved a $(2+\epsilon)$-approximation in $O(k^2 \epsilon^{-1} \log{\Delta})$ amortized update time, where $\Delta$ is the aspect ratio of the metric space. 
Later on, the amortized update time was improved to $O(k \epsilon^{-1} \textrm{poly} \log (n, \Delta))$ by Bateni et al.~\cite{BateniEFHJMW23}. 
This is almost optimal considering that even in the static setting, any algorithm for $k$-clustering problems (including $k$-center, $k$-median, and $k$-means) on point sets in arbitrary metric spaces that achieves any non-trivial approximation, must make at least $\Omega(nk)$ distance queries, and in turn must take $\Omega(nk)$ time~\cite{BateniEFHJMW23}.
These results spurred several follow-up works that studied other point sets-based clustering objectives such as $k$-means~\cite{Cohen-AddadHPSS19, HenzingerK20}, $k$-median~\cite{GuoKLX21}, facility location~\cite{GoranciHL18, GuoKLX20, GoranciHLSS21, BhattacharyaLP22}, and sum-of-radii~\cite{HenzingerLM20} in the dynamic setting. 

An important case of $k$-center clustering is when the input metric is induced by a graph $G$ on $n$ vertices and $m$ edges. Naturally, any $k$-center algorithm that works with points in arbitrary metric spaces, can be applied on top of the graphical metric obtained by computing all-pairs shortest paths in $G$. However, the latter leads to slow running times, especially since it could make a sparse graph $G$ very dense. In the static setting, Thorup~\cite{Thorup04} gave a faster algorithm for the $k$-center problem in the graph setting, achieving a $(2+\epsilon)$-approximation in $\tilde{O}(m\epsilon^{-1})$ time, where $\tilde{O}$ hides polylogarithmic factors in $n$ and in the maximum edge-weight of the graph. This result was recently revisited by the work of Eppstein et al.~\cite{eppstein2015approximate} and even more recently by Abboud et al.~\cite{AbboudCLM23} who gave a refined and simpler algorithm for $k$-center on graphs.

We note that graph clustering has also received attention in the machine learning community, albeit for the closely related objective of $k$-means~\cite{RattiganMJ07}. 
They observe the computational challenges involving graphs (see also~\cite[Section~2.3]{Aggarwal2010}) and specifically the output sensitivity due to distance changes caused by edge updates. 

Motivated by these developments, we study the fundamental problem of \emph{dynamic $k$-center on graphs}. In comparison to the model with dynamic point sets in arbitrary metric spaces, we remark that the model with dynamic graphs is more challenging since (i) there is no guarantee of having oracle access to all-pairs shortest paths distances, and (ii) a single edge update may have a \emph{global effect} on the underlying graph metric, forcing a large number of vertex pairs to change their shortest path distance. This is also why we cannot use other black-box approaches such as a distance oracle of the metric completion of the graph.

To that end, we ask the natural question of to what extent one can leverage the structure of graphical $k$-center in the context of obtaining faster algorithms for dynamic $k$-center on graphs:

\begin{center}
    \emph{Are there efficient algorithms for $k$-center on graphs undergoing edge updates?}
\end{center}

\subsection{Our Contribution} 

In this paper, we answer the question in the affirmative. 
Our first contribution is a fully dynamic $k$-center algorithm that follows from prior work using a surprisingly simple trick.
\begin{restatable}{theorem}{fullydynamickcenter}\label{thm:fully-dynamic approx}
Given a weighted undirected graph $G=(V,E,w)$ subject to edge updates, an integer parameter $k \geq 1$, and a positive constant parameter $\epsilon\le 1/2$, there are two fully dynamic algorithms for the $k$-center problem on graphs, that maintain a $(2+\epsilon)$-approximation with the following guarantees (based on the current value of the matrix multiplication exponent):
\begin{enumerate}[topsep=0pt,itemsep=-1ex,partopsep=1ex,parsep=1ex]
    \item Deterministic algorithm with $O(kn^{1.529}\epsilon^{-2})$ worst-case update time, if $G$ has uniform weights; 
    \item Randomized algorithm, against an adaptive adversary, with $O(kn^{1.823} \epsilon^{-2})$ worst-case update time, if $G$ has general weights.
\end{enumerate}
Both algorithms have preprocessing time $O(n^{2.373} \epsilon^{-2} \log \epsilon^{-1})$.
\end{restatable}

Note that our update time bounds match up to an $ \tilde O (k) $ factor, those of the state-of-the-art fully-dynamic single-source distance approximation algorithms with multiplicative error $(1+\epsilon)$~\cite{BrandN19,BFN22}.
Matching this bound is a natural goalpost for dynamic $k$-center algorithms maintaining the $(1 + \epsilon)$-approximate distance of each vertex to its closest center, because such distance approximations are sufficient to return a $(2 + \epsilon)$-approximation for graph diameter when $ k = 1 $ and the fastest known approach for this is to use a dynamic single-source distance approximation algorithm.
Our algorithms---and to the best of our knowledge all related dynamic $k$-center algorithms on general metrics---do have this desirable property of maintaining the $(1 + \epsilon)$-approximate distance of each vertex to its closest center.
Moreover, the previous result is a reduction to the problem of maintaining $k$-source approximate shortest paths in a fully dynamic setting; hence any improvement on the shortest paths algorithms directly improves the running time of our algorithms as well.

The above suggests that in order to achieve faster running times, we need to consider \emph{partially dynamic algorithms} for the $k$-center problem on graphs, where edge updates are restricted to only edge insertions or edge deletions. In particular, the insertions-only algorithms (also known as the \emph{incremental setting}) in the context of clustering are particularly well-motivated from a practical viewpoint. For example, real-world graphs such as co-authorship networks are incremental since the fact that two scientists co-authoring a research paper (almost) never changes over time. Our main result regarding the incremental setting is the following. 

\begin{restatable}{theorem}{incrkcentfappr}\label{th:incr_kcent_4appr}
    Given a weighted undirected graph $G = (V, E, w)$ subject to edge insertions, an integer parameter $k \geq 1$, and a positive constant parameter $\epsilon<1$, there is
    a randomized incremental $(4 + \epsilon)$-approximation algorithm for the $k$-center problem on graphs, which w.h.p.~is correct and w.h.p.~has $kn^{o(1)}$ amortized update time.
\end{restatable}

To complete the picture of partially dynamic algorithms, we also study the $k$-center problem on graphs undergoing edge deletions only, known as the \emph{decremental} setting. Here, we obtain an algorithm that achieves a tight $(2+\epsilon)$ approximation ratio. 

\begin{restatable}{theorem}{deckcenttappr} \label{th:dec_kcent_2appr}
    Given a weighted undirected graph $G = (V, E, w)$ subject to edge deletions, an integer parameter $k \geq 1$,
    and a positive constant parameter $\epsilon < 1$, there is a deterministic decremental 
    $(2 + \epsilon)$-approximation  algorithm for the $k$-center problem on graphs, with
    $kn^{o(1)}$ amortized update time over a sequence of $\Theta(m)$ updates.
\end{restatable}

We note that the $n^{o(1)}$ factors in the running time are also due to using partially dynamic approximate single-source shortest paths (SSSP) algorithms, which is inherent in our bounds based on similar reasoning as in the fully dynamic setting.

\paragraph{Outline.}
In Section~\ref{sec:overview}, we give an overview of our algorithm and also discuss the main challenges we face in dynamic graphs, as opposed to point sets.
In Section~\ref{sec:reduction}, we review a well-known reduction that relates the $k$-center problem to finding a maximal independent set on a graph. This reduction is fundamental to our partially dynamic algorithms.
Section~\ref{sec:incremental} presents our primary technical contribution, showcasing the incremental algorithm of Theorem~\ref{th:incr_kcent_4appr}.
Section~\ref{sec:decremental} completes the partially dynamic picture by providing the decremental algorithm of Theorem~\ref{th:dec_kcent_2appr}.
In Section~\ref{sec:fullydynamic}, we explore the fully-dynamic setting, in which we use a different type of algorithm to prove Theorem~\ref{thm:fully-dynamic approx}. Unlike the reduction presented in Section~\ref{sec:reduction}, our approach here is based on the greedy algorithm of Gonzalez~\cite{Gonzalez85}.
In addition to the set of centers, we can also answer other natural queries, such as the corresponding center for each vertex. We briefly discuss this in Appendix~\ref{apx:queries}.

\section{Technical Overview}\label{sec:overview}

In this section, we give a high-level overview of our algorithms and discuss several technical challenges that we need to handle for dynamically maintaining a $k$-center solution on graphs rather than on point sets.

\paragraph{Reduction to $k$-Bounded Maximal Independent Set on Threshold Graphs.} We start by reviewing a known reduction from $2$-approximate $k$-center to $k$-bounded maximal independent set (MIS)~\cite{HochbaumS86}. This reduction is also the basis of some of the fully dynamic $k$-center algorithms on point sets in arbitrary metric spaces~\cite{BateniEFHJMW23, chan2018fully}. The idea for obtaining a $2$-approximation is to guess the optimal value $R^*$ of the $k$-center instance via a binary search, and return
any maximal distance-$2R^*$ independent set $M$. It can be shown that $M$ must be of size at most $k$. Recall that $M$ is a maximal subset of vertices such that no two vertices $u,v \in M$ are within distance $2R^*$ of each other. The vertices of $M$ then correspond to the centers of the $k$-center instance. 

For utilizing this idea in the dynamic setting, rather than guessing the value of $R^*$, we maintain the MIS on each \textit{$r$-threshold graph} $G_r$, for every \emph{distance range} $r \leq (1+\epsilon)^{i}$, where $i \in [0, \log_{1+\epsilon} (nW)]$ and $W$ is the maximum edge-weight of the graph. 
Here, by $r$-threshold graph we mean a graph such that there is an edge between two vertices if and only if they are within distance $r$.
This general framework has been used in the dynamic $k$-center algorithms on point sets \cite{BateniEFHJMW23, chan2018fully}. Actually, they  maintain a relaxation of the MIS,
called a \emph{$k$-bounded MIS}. The observation is that it is sufficient to either return an MIS of size at most $k$ on each $r$-threshold graph $G_r$, or to simply report that there is an independent set of size at least $k+1$ as a witness that the distance range $r$ is not the correct guess.

\paragraph{Technical Challenges in Graphs vs Point Sets.}
As discussed before there are several important technical differences between the graph and point sets setting. The first difference is that, unlike the point sets setting, in graphs we do not have direct access to distances. Thus, we also need to maintain the appropriate distances simultaneously while the dynamic MIS is modified on the $r$-threshold graphs.
For this, we would like to combine a dynamic MIS maintenance algorithm with partially dynamic (approximate) shortest paths algorithms. At a high-level, our goal is to maintain a $k$-bounded MIS $M$ dynamically on each $r$-threshold graph $G_r$ and at the same time maintain a dynamic (approximate) SSSP algorithm from a super-source which is connected to all the vertices in the dynamic set $M$. Hence the efficiency of the algorithm will depend on the number of times the set $M$ is modified over all the updates. Equivalently, the efficiency will depend on the number of times the dynamic SSSP algorithm is restarted.

In both point sets and graph setting, we need to bound the \textit{recourse}, where by recourse we mean the number of times a new center is introduced by the algorithm. However, we argue that a \textit{stronger recourse guarantee} is needed for graphs. Recall that another important difference between these two settings is that in the point sets setting adding or removing points has a more local impact, whereas in a graph an update may impact the distances between many vertices. 
In other words, in the graph setting an edge update may \textit{distort} the metric itself. This difference in graphs together with the fact that the efficiency will depend on the number of times the dynamic set $M$ is modified/dynamic SSSP is restarted, requires our algorithm to have an overall recourse guarantee, as opposed to an amortized one that suffices for the point sets setting. 
Thus an amortized recourse of $\tilde{O}(k)$ \emph{per update} is not enough, and we need the stronger guarantee that the recourse is $\tilde{O}(k)$ \emph{over all updates}. 
More concretely, this total bound on recourse will let us argue that in total we need to re-initialize a dynamic $(1+\epsilon)$-SSSP algorithm $\tilde{O}(k)$ times from each center, and an amortized guarantee would not be enough for getting our desired update bound. We will see that maintaining this stronger recourse guarantee is more challenging in the incremental setting than in the decremental setting. 

Note that these types of recourse guarantees have been studied in point sets from arbitrary metric spaces under the name \emph{consistent} clustering~\cite{LattanziV17,FichtenbergerLNS21,LackiHGRJ23}.
However, we would like to emphasize that the known sublinear bounds on the total recourse in the point sets setting do not carry over to the graph setting.

\paragraph{Decremental $k$-Center on Graphs.} 
We can obtain a decremental $(2+\epsilon)$-approximation algorithm for the $k$-center problem on graphs, by maintaining a decremental $(1+\epsilon)$-SSSP algorithm from a
super-source which is connected to all the centers in each of the $O(\log nW)$ $r$-threshold graphs. Bounding the recourse in the decremental setting is relatively straightforward based on the following observation. Whenever a new center forms a cluster due to a distance increase in a given $r$-threshold graph, it stays disjoint of other clusters throughout the algorithm and thus it stays a valid center. Furthermore, as soon as we get more than $k$ centers, we move to the next distance range and so, the recourse is upper bounded by $k$ on each $r$-threshold graph and by $O(k \log nW)$ overall. Hence the $(1+\epsilon)$-SSSP algorithm
is restarted at most $O(k \log nW)$ times in total. This combined with the time needed for maintaining partially dynamic $(1+\epsilon)$-approximate SSSP leads to our desired $kn^{o(1)}$ amortized update time.

\subsection{Incremental $k$-Center on Graphs}
\paragraph{Incremental Low Recourse Ruling Sets.} 
Bounding the recourse in the incremental setting is more challenging compared to the decremental setting, for the following reason. After an edge insertion in the input graph, a center $c_1$ of a cluster
may come within distance $r$ of an existing center $c_2$ of another cluster. 
In turn, this means that the two vertices $c_1$ and $c_2$ become neighbors in the $r$-threshold graph $G_r$.
We cannot simply merge the corresponding clusters in some way and still maintain a $2$-approximation, as some vertices in such a merged cluster would go beyond 
the desired distance range after each update. Hence we need a new technical idea to keep the recourse low. The idea is to 
maintain a small (i.e., of size $\tilde{O}(k)$) \textit{dominating set} $S$ on $G_r$ such that, at a high-level, maintaining a maximal independent set on $G_r[S]$ will give us an \textit{approximate} maximal independent set on $G_r$. More formally, by maintaining a $k$-bounded MIS $M$ on the dominating set $S$ in $G_r$, we can show that $M$ is also a $k$-bounded $(2,2)$-ruling set\footnote{While our $k$-center algorithms work for weighted graphs, the ruling set subroutines always perform on unweighted graphs regardless on the input to the $k$-center problem.} on $G_r$. That is, a subset $M$ of vertices of size at most $k$ such that: (i) the distance in $G_r$ between any pair of vertices in $M$ is at least $2$, and (ii) for each vertex in $V$ there exists a vertex in $M$ within distance $2$ in $G_r$.
Introducing this small dominating set allows us to maintain a dynamic $k$-bounded maximal independent set $M$ on the smaller subgraph $G_r[S]$ more efficiently, at the cost of losing a factor $2$ in the approximation due to the fact that $M$ is only a $(2,2)$-ruling set on $G_r$. 
For maintaining such a dominating set, we use a recursive algorithm that maintains a union of hitting sets on a sequence of sparsified subgraphs of $G_r$. The hitting sets are obtained by a standard sampling procedure on the subgraphs corresponding to each recursive call. Informally, the sampling rate of the hitting sets is tuned depending on the densities of these subgraphs and the recursion continues until the remaining set of low degree vertices is sufficiently small. 
Since we have an incremental graph, the set of \textit{low degree vertices defined based on a specific degree threshold} that are not covered by the hitting sets shrinks over time. This together with observations regarding the sampling, the degrees, and a property of the $k$-center problem allows us to bound the recursion depth by $O(\log n)$.

\paragraph{Challenges of Working with the $r$-Threshold Graphs.} 
The high-level idea described above will give us an algorithm for maintaining a $k$-bounded $(2,2)$-ruling set on an incremental graph in $\tilde{O}(k)$ amortized update time with an overall recourse of $\tilde{O}(k)$ \antonis{modify it}. Similar to the decremental algorithm, our goal is to maintain such a ruling set on all $r$-threshold graphs. This, combined with an incremental $(1+\epsilon)$-SSSP algorithm, will lead to an incremental $(4+\epsilon)$-approximation algorithm for the $k$-center problem on graphs. The main remaining challenge is that an edge insertion into the input graph $G$ could lead to many edge insertions in the $r$-threshold graph. In turn, the density of the $r$-threshold graphs could be $n^2 \geq m$. To overcome this challenge, we do not explicitly store the $r$-threshold graphs.
Instead, we utilize the construction of the dominating set on each $r$-threshold graph $G_r$ together with the fact
that the dominating set for each $G_r$ is of small size, to ensure that only \textit{relevant} edges
are processed. That is, edges that either participate in the construction of the dominating set or those that cause a conflict.
Overall, by bounding the number of candidate centers and taking advantage of the construction of the dominating sets, 
we ensure that $\tilde{O}(k)$ incremental $(1+\epsilon)$-SSSP algorithms are re-initialized in total, 
and the \textit{relevant} information is maintained. In turn, this leads to an amortized update time of $kn^{o(1)}$.

\section{Preliminaries}\label{sec:preliminaries}
    \paragraph{Graphs.}
    Consider a weighted undirected graph $G = (V, E, w)$. We denote by $n = |V|$ the number of vertices, by $m = |E|$ the number
    of edges, and by $W$ the maximum weight of an edge. 
    Without loss of generality (w.l.o.g.), we assume that the minimum edge weight is equal to 1.
    Moreover, we assume that $W$ is bounded by a polynomial in $n$ (i.e., $W=O(\text{poly}(n))$).
    
    For any two vertices $u, v \in V$, the \emph{distance $d_G(u, v)$ between $u$ and $v$} is the 
    length of a shortest path from $u$ to $v$ in $G$. For a fixed subset of vertices $S \subseteq V$ and a vertex $v \in V$, the \emph{distance $d_G(v, S)$ 
    between $v$ and $S$} is equal to $\min_{u \in S} d_G(v, u)$, namely the distance from $v$ to its closest vertex in $S$. 
    For a vertex $v \in V$, we denote by $N_G(v)$ the set of neighbors of $v$ in $G$, and by $N_G[v] := N_G(v) \cup \{v\}$ the \emph{closed neighborhood} of $v$ in $G$.
    A \emph{subgraph} $H$ of a graph $G$ is a graph whose vertex set and edge set are subsets
    of the vertex set and edge set of $G$ respectively. 
    An \textit{edge-subgraph} of $G$ is a graph whose vertex set is the same as the vertex set of $G$ and whose edge set is a subset of the edge set of $G$.
    For a subset of vertices $S \subseteq V$,
    the \emph{induced subgraph $G[S]$} is the graph with vertex set $S$, whose edge set consists
    of all edges in $E$ that have both endpoints in $S$. We also say that $G[S]$ is the
    subgraph induced by $S$. For a graph $H$, we denote by $V(H)$ the vertex set of $H$, and by $E(H)$ the edge set of $H$.

    Consider now an unweighted undirected graph $G = (V, E)$. 
    A \emph{distance-$\alpha$ independent set} $M$ is a subset of vertices such that the distance between any two vertices in $M$
    is strictly more than $\alpha$. An \emph{independent set (IS)} is a distance-$1$ independent set.
    An \emph{\abrs} is a subset of vertices $M \subseteq V$ such that the distance between any two vertices in $M$ is at least $\alpha$,
    and the distance between any vertex in $V$ and its closest vertex in $M$ is at most $\beta$.
    A \emph{maximal independent set (MIS)} is a $(2, 1)$-ruling set.

    \paragraph{Dynamic Setting.}
    In the dynamic setting, the input graph $G$ is subject to edge updates. Namely, edges can be inserted into $G$ (edge insertions) and/or edges can be removed from $G$ (edge deletions).
    % Observe that an update which changes the weight of an edge can be simulated by two updates, where the first update deletes
    % the corresponding edge and the second update re-inserts the edge with the new weight. \antonis{We can handle weight updates here, right?}
    A \emph{fully dynamic algorithm} is able to process both types of edge updates (i.e., edge insertions and edge deletions), while
    a \emph{partially dynamic algorithm} is able to process only one type of edge updates
    (i.e., either edge insertions or edge deletions). 
    In particular, an \emph{incremental algorithm} can process only edge insertions and a \emph{decremental algorithm} can process only edge deletions.

    In our incremental algorithms, we assume that the updates are performed by an \emph{oblivious adversary} who fixes the sequence of updates before the algorithm starts. Namely,
    the adversary cannot adapt the updates based on the choices of the algorithm during the execution.
    This is as opposed to an \textit{adaptive adversary}, that instead we consider in the decremental and fully-dynamic settings. A dynamic algorithm has \emph{amortized update time} $u(n, m)$ if its total time spent for processing any sequence of $\ell$ updates is bounded by $\ell \cdot u(n, m)$.
    
    In the incremental setting, let $M$ be an independent set in $G$. 
    Then for an edge insertion $(u, v)$ in $G$, we say that the edge $(u, v)$ causes a \emph{conflict} in $G$ when both of its endpoints $u$ and $v$ belong to $M$ before the update.
     
    % In partially dynamic graphs we consider $\ell = O(n^2)$.

\subsection{$k$-Center on Graphs}  
The $k$-center problem on graphs is defined formally as follows.
\begin{Definition}[$k$-center on graph]
    Given a weighted undirected graph $G = (V, E, w)$ and an integer $k \geq 1$, the goal is to output a subset of vertices $S \subseteq V$ of size at most $k$,
    such that the value $\max_{v \in V} d_G(v, S)$ is minimized.
\end{Definition}

Consider a \emph{$k$-center instance $\left(G = (V, E, w), k\right)$}, which is the pair of the given input graph $G$ and the integer $k$.
For each choice of $S \subseteq V$, we define the radius $r := \max_{v \in V} d_G(v, S)$. 
The vertices of $S$ are also called \emph{centers}.
For a fixed $S$ with radius $r$, we define a \emph{cluster} for every $c \in S$ containing all vertices within distance $r$ from the center $c$.
We denote by $R^* := \min_{|S| \leq k} \max_{v \in V} d_G(v, S)$ the optimal radius of the given instance,
and by $S^*$ any subset with radius $R^*$ (i.e., $R^* = \max_{v \in V} d_G(v, S^*)$). For completeness we also discuss how we may be interested in answering other type of queries in Appendix \ref{apx:queries}. In the dynamic setting, the input graph of the $k$-center instance is subject to edge updates.

\subsection{Partially Dynamic Shortest Paths Algorithms}
Through the paper, we heavily make use of the existing partially dynamic $(1+\epsilon)$-approximate single-source shortest paths (SSSP) algorithms. In the decremental setting, we can use a deterministic algorithm.
\begin{theorem}[Decremental $(1+\epsilon)$-SSSP, \cite{bernstein2021deterministic}]\label{th:decr_appr_sssp}
    Given a weighted undirected graph $G = (V, E, w)$ subject to edge deletions, a source $s \in V$, and a constant $\epsilon\in(0,1)$, there is a deterministic algorithm that maintains $(1 + \epsilon)$-approximate shortest paths from $s$
    in total update time $m^{1+o(1)}$.
\end{theorem}

\noindent
In the incremental setting, we can use the following randomized partially dynamic algorithm.
\begin{theorem}[Incremental $(1+\epsilon)$-SSSP, \cite{HKN2014hopsets, chechik2018, LackiN22}] \label{th:incr_appr_sssp}
    Given a weighted undirected graph $G = (V, E, w)$ subject to edge insertions, a source $s \in V$, and a constant $\epsilon\in(0,1)$, there is a randomized algorithm (against an oblivious adversary) that maintains $(1 + \epsilon)$-approximate shortest paths from $s$ in total update time $m^{1+o(1)}$.
\end{theorem}
The incremental $(1+\epsilon)$-SSSP algorithm is not explicitly stated but follows from similar algorithms as the decremental setting such as \cite{HKN2014hopsets, chechik2018, LackiN22}. In Appendix~\ref{app:incremental}, we give a brief sketch of how one can adapt these results to the incremental setting, but the details of this algorithm are beyond the scope of this paper.

\section{Reduction from $2$-Approximate $k$-Center to $k$-Bounded Ruling Set} \label{sec:reduction}
It is well-know that the $k$-center problem can be reduced to the problem of finding an MIS on a graph. This reduction was first given by Hochbaum and Schmoys~\cite{HochbaumS86}, 
in order to get a $2$-approximation algorithm, and also it has been used by \cite{chan2018fully, BateniEFHJMW23} 
for the fully dynamic $k$-center problem on point sets in arbitrary metric spaces.
In particular, it is sufficient to solve a weaker version of the MIS problem,
where we only need to return an MIS of size at most $k$, or report that there is an independent set of size at least $k + 1$.
Formally, we define this problem based on an $(\alpha, \beta)$-ruling-set, as follows.
A similar definition was also given in \cite{BateniEFHJMW23} for the MIS. Recall that an MIS is a $(2, 1)$-ruling set. 

\begin{Definition}[$k$-bounded $(\alpha, \beta)$-ruling set problem]\label{def:kbabrs}
    Given an unweighted undirected graph $G = (V, E)$, an integer $k \geq 1$, and parameters $\alpha, \beta$ such that $\beta \geq \alpha - 1 \geq 0$, the \kbabrs problem asks to either return an \abrs of size at most $k$, or to report that there is a distance-$(\alpha-1)$ independent set of size at least $k + 1$.
\end{Definition}

\noindent
The reduction solves the \kbabrs problem on the following type of graphs.

\begin{Definition}[$r$-threshold graph]\label{def:threshold_graph}
Given a weighted graph $G = (V, E, w)$ and a parameter $r > 0$, 
the \emph{$r$-threshold graph} $G_r = (V, E_r)$ is defined as the graph with vertex set $V$ and edge set $E_r = \{(u,v) \in V \times V : d_G(u,v) \leq r \}$.  
\end{Definition}
\noindent
In other words, the $r$-threshold graph $G_r$ connects all pairs of vertices that are within distance $r$ in $G$. Observe that the $r$-threshold graph $G_r$ is unweighted.

The next lemma is an adjustment of the reduction of Hochbaum and Schmoys~\cite{HochbaumS86} to Definition~\ref{def:kbabrs}.
The proof is deferred to Appendix~\ref{apx:reduction proofs}.

\begin{restatable}{lemma}{tbapproxGr}\label{lem:2b-approx_Gr}
    Consider a $k$-center instance $\left(G = (V, E, w), k\right)$, and a positive constant parameter $\epsilon$. 
    Then by running a \kbtbrs algorithm on the $r$-threshold graph $G_r$, for each $r \in \{(1 + \epsilon)^i \mid (1 + \epsilon)^i \leq nW, i \in \mathbb{N}\}$,
    we can find a $2\beta(1 + \epsilon)$-approximate solution for the $k$-center instance.
\end{restatable}

For the sake of efficiency, in the dynamic setting we do not handle $r$-threshold graphs,
but rather an approximation of them. For this reason, we generalize the previous lemma as follows.
\begin{lemma}\label{lem:2b_Gr_H}
    Consider a $k$-center instance $\left(G = (V, E, w), k\right)$, constant positive parameters $\epsilon,\epsilon'$, $r \geq 0$, and $\beta\ge 1$. 
    Let $r' := (1+\epsilon') r$ and consider the threshold graphs $G_r$ and $G_{r'}$.
    Assume that there is an algorithm $\mathcal{A}$ such that, given $G, r, \epsilon',\beta$, 
    \begin{itemize}[topsep=0pt,itemsep=-1ex,partopsep=1ex,parsep=1ex]
        \item either reports that there is an independent set in $G_r$ of size at least $k + 1$, 
        \item or runs a \kbtbrs algorithm $\mathcal{B}$ on an edge-subgraph $H$ of $G_{r'}$ with the following condition: whenever $\mathcal{B}$ reports that there is an independent set 
        in $H$ of size at least $k + 1$, then there is an independent set in $G_r$ of size at 
        least $k + 1$.
    \end{itemize}
    Then, by running $\mathcal{A}$ with input $G, r, \epsilon',\beta$, for each $r \in \{(1 + \epsilon)^i \mid (1 + \epsilon)^i \leq nW, i \in \mathbb{N}\}$, we can find a $2\beta(1 + \epsilon)(1 + \epsilon')$-approximate solution for the $k$-center instance.
\end{lemma}
\noindent
As already stated, the previous lemma is a generalization of Lemma~\ref{lem:2b-approx_Gr}. 
In fact, observe that in the definition of a \kbtbrs problem, we are allowed to report that there is an independent set of size at least $k + 1$. 
Thus by setting $H = G_r$ and $\epsilon' = 0$ in Lemma~\ref{lem:2b_Gr_H}, we get Lemma~\ref{lem:2b-approx_Gr} as a corollary.

Before proving Lemma~\ref{lem:2b_Gr_H}, we state two auxiliary results that will be useful. Their proofs are deferred to Appendix~\ref{apx:reduction proofs}.

\begin{restatable}{lemma}{kmisrtR}\label{lem:k-MIS_r>=2R^*}
    Consider a $k$-center instance $\left(G = (V, E, w), k\right)$ with optimal radius $R^*$. 
    Then for each $r \geq 2R^*$ and for every $\beta \geq 1$, it holds that every \tbrs in the $r$-threshold graph $G_r$ is of size at most $k$.
\end{restatable} 

\begin{restatable}{observation}{noIStRGr}\label{obs:noIStRGr}
Consider a $k$-center instance $\left(G = (V, E, w), k\right)$ with optimal radius $R^*$, and let
$G_r$ be the $r$-threshold graph where $r=2R^*$.
Then, there is no independent set in $G_r$ of size at least $k + 1$.
\end{restatable}

\noindent
We proceed now with the proof of Lemma~\ref{lem:2b_Gr_H}.

\begin{proof}[Proof of Lemma~\ref{lem:2b_Gr_H}]
    Let $\hat{r}$ be the smallest $r \in \{(1 + \epsilon)^i \mid (1 + \epsilon)^i \leq nW, i \in \mathbb{N}\}$ such that algorithm $\mathcal{A}$ returns a \tbrs $M_r$ of size at most $k$ in an edge-subgraph $H$ of $G_{r'}$, where $r' = (1+\epsilon')r$. 
    Also, let $\hat{r}' = (1+\epsilon')\hat{r}$ and let $S := M_{\hat{r}}$ be the solution we return for the $k$-center instance.

    Since $H$ is a subgraph of $G_{\hat{r}'}$, 
    then for every edge $(u, v) \in E(H)$, the distance between $u$ and $v$ in $G$ is at most $\hat{r}'$. Hence, as $M_{\hat{r}}$ is a \tbrs in $H$, then every vertex is within distance 
    $\beta \hat{r}'$ from its closest center in $G$. Thus, the returned solution $S$ has radius
    at most $\beta  \hat{r}'$. 
    
    We show now that $\hat{r}$ 
    is at most $2(1+\epsilon)$ times larger than $R^*$.
    Based on Observation~\ref{obs:noIStRGr}, for the fixed choice of $r = 2R^*$,
    algorithm $\mathcal{A}$ always returns a \tbrs $M_r$ in $H$ of size at most $k$.
    By definition of $\hat{r}$, and since the possible values of $r$ are powers of $(1 + \epsilon)$,
    we have that $\hat{r} \leq 2(1 + \epsilon)R^*$.
    Therefore, the radius of the returned solution $S$ is at most $2\beta(1+\epsilon)
    (1+\epsilon')R^*$.
\end{proof}

\section{Incremental $k$-Center on Graphs}\label{sec:incremental}
In the incremental setting, the input graph of the $k$-center instance is subject to edge insertions.
We start by recalling the concept of \textit{dominating set}, which we will exploit throughout this section.

\begin{Definition}[Dominating set]\label{def:kern_set}
    Given an unweighted undirected graph $G = (V, E)$, a dominating set $S \subseteq V$ in $G$ is a subset of vertices such that each vertex of $G$ is either in $S$ or has a neighbor in $S$.
\end{Definition}

\begin{observation} \label{obs:rs_beta_beta1}
    Consider an unweighted undirected graph $G=(V,E)$. Let $S$ be a dominating set in $G$, and $M$ be an $(\alpha, \beta)$-ruling set in $G[S]$. Then $M$ is an $(\alpha, \beta+1)$-ruling set in $G$.
\end{observation}

In this section, we first develop an incremental algorithm for the \kbttrs problem by finding a small dominating set $S$ and maintaining a \kbtors in $G[S]$, as Observation~\ref{obs:rs_beta_beta1} suggests. The idea is to use the reduction of Lemma~\ref{lem:2b-approx_Gr} with this algorithm, in order to solve the incremental $k$-center problem. 
In the reduction though, notice that we need to maintain an incremental \kbttrs algorithm on $r$-threshold graphs, which is more challenging. 
To that end, in Section~\ref{sec:incr_kbtt_Gr} we develop an efficient incremental \kbttrs algorithm that works on approximate versions of $r$-threshold graphs.
Finally, we apply Lemma~\ref{lem:2b_Gr_H} instead of Lemma~\ref{lem:2b-approx_Gr}, to obtain the incremental $k$-center algorithm.

\subsection{Incremental $k$-Bounded $(2, 2)$-Ruling Set Algorithm} \label{sec:2_k_MIS}
We begin by describing how to detect a small dominating set $S$ on an incremental graph
$G$, and maintain a $k$-bounded $(2,1)$-ruling set on the subgraph induced by $S$. In particular, we prove the following theorem.

\begin{restatable}{theorem}{reskbmis}\label{th:res_kbmis}
    Given an unweighted undirected graph $G = (V, E)$ subject to edge insertions, and an integer $k \geq 1$, there is a randomized incremental algorithm which:
    \begin{itemize}[topsep=0pt,itemsep=-1ex,partopsep=1ex,parsep=1ex]
        \item either reports that there is an independent set in $G$ of size at least $k + 1$, and this is correct w.h.p.,
        \item or finds a \emph{dominating set} $S$ of size $\tilde{O}(k)$ in $G$ and maintains a \kbtors in $G[S]$.
    \end{itemize}
\end{restatable}

Notice that based on Observation~\ref{obs:rs_beta_beta1} and the definitions of dominating set and \kbabrs problem (i.e., Definition~\ref{def:kern_set} and
Definition~\ref{def:kbabrs}), the algorithm of Theorem~\ref{th:res_kbmis} solves
the incremental \kbttrs problem in $G$.
 Before describing the algorithm we review two existing algorithms tools.
First tool is the following folklore hitting set claim (e.g., see \cite{AingworthCIM99}, also widely used in decremental settings against an oblivious adversary).
\begin{lemma} \label{lem:sampling}
    Given a graph $G = (V, E)$ and a threshold $\gamma \ge 1$, let $S$ be the set obtained by sampling each vertex independently with probability $\min(c\ln(n) / \gamma, 1)$, for a constant $c>1$. Then, with probability at least $1 - n^{-(c-1)}$, every vertex of degree more than $\gamma$ has at least one neighbor in $S$.
\end{lemma}
\noindent
As noted, e.g., in~\cite{rodittyzwick2012dynamic}, even though Lemma~\ref{lem:sampling} refers to a static graph, it is easy to see that it holds for partially dynamic graphs. Since we are assuming an oblivious adversary, the choice of the random set $S$ is independent of the graph. This and the fact that we have at most $ O(n^2)$ versions of the graph in the incremental setting, let us bound the overall probability via a straightforward union bound, and the failure probability is at most $n^{-(c-3)}$.

Second tool, is a fully dynamic \kbtors algorithm
with the following guarantees. This
algorithm is a trivial extension of any fully dynamic MIS algorithm that returns explicitly
the MIS. For that reason, we can either use the MIS algorithm of Behnezhad et al.~\cite{behnezhad2019fully},
or the algorithm of Chechik and Zhang \cite{chechik2019fully}.
 
\begin{theorem} \label{th:fully_dyn_mis}
    Given a graph $G = (V, E)$ subject to edge updates, there is a fully dynamic 
    \kbtors algorithm with $\tilde{O}(1)$ amortized update time.
\end{theorem}
\begin{proof}
    The algorithm of Behnezhad et al.~\cite{behnezhad2019fully} maintains
    an MIS $M$ under edge updates, in $\tilde{O}(1)$ amortized update time. Recall that an MIS is a \tors.
    Thus at any moment, if the size of $M$ is at least $k + 1$, we report
    that there is an independent set in $G$ of size at least $k + 1$, otherwise we return the set $M$.
\end{proof}

\subsubsection{Overview of the Algorithm}
A pseudocode of the algorithm of Theorem \ref{th:res_kbmis} is provided in Algorithm~\ref{alg:2_k_MIS}.
The algorithm consists of two phases. 
Roughly speaking, the first phase either detects a dominating set $S$ or reports that there is an independent set in $G$ of size at least $k + 1$. 
The second phase starts when such a dominating set $S$ is detected and is only responsible for maintaining an incremental \kbtors in $G[S]$.

\begin{algorithm}[ht!]
\DontPrintSemicolon
\caption{\textsc{\kbttrs}}
\label{alg:2_k_MIS}

\setcounter{AlgoLine}{0}
\SetAlgoLined

\SetKwProg{Fn}{Procedure}{:}{\KwRet}
\SetKwFunction{FSample}{SampleSL}
\SetKwFunction{FKBRS}{$k$-Bounded-Ruling-Set}
\SetKwFunction{FUpdate}{Insert}

\tcp{In the preprocessing $i = 0$, $L_0 = V$, and $\FKBRS{}$ is called with no edge (note that Line~16 where the edge is actually used cannot be reached during the preprocessing)}
\tcp{The index $i$ and the sets $L_i,S_i$ for every $i$ are global}

\vspace{1em}

\Fn{\FKBRS{$u,v$}} {
    \If(\tcp*[h]{first phase}){$|L_i| > 4k$} {
        \If(\tcp*[h]{recursive sampling}){$i= 0$ or $|L_i|\le \frac{|L_{i-1}|}{2}$}{
            $i \gets i+1$
        
            $\gamma_i \gets \frac{|L_{i-1}|}{2k} - 1$
        
            $S_i \gets$ sample vertices of $L_{i-1}$ independently with prob.~$\min(10\ln(n) / \gamma_i, 1)$
            \tcp*{Lemma~\ref{lem:sampling}}
            
            $L_i \gets \{x \in L_{i-1} : N_{G[L_{i-1}]}(x) \cap S_i = \emptyset\}$
    
            \FKBRS{u,v}
        }
        \Else{
            \textbf{report} there is an independent set in $G$ of size at least $k+1$
            \tcp*{Lemma~\ref{lem:kvert_nomis}}
        }
    }
    \Else(\tcp*[h]{second phase}) {
        
        \If(\tcp*[h]{$\mathcal{B}$ is dynamic \kbtors algorithm (Theorem~\ref{th:fully_dyn_mis})}){$\mathcal{B}$ is not initialized}{
            $d \gets i$
            
            $S \gets \bigcup_{j=1}^d S_j \cup L_d$
            \tcp*{$S$ is a dominating set (Lemma~\ref{lem:S_kern_G})}

            $\mathcal{B}.initialize(G[S])$
        }
        \ElseIf{$u\in S$ and $v \in S$}{
            $\mathcal{B}.update(G[S], u,v)$
        }
    }
}

\vspace{1em}

\Fn{\FUpdate{$u, v$}} {
    $G \gets (V,\, E \cup \{u,v\})$
    
    \If {$u \in L_i$ and $\exists j \leq i$ such that $v \in S_j\,$ (resp., $v \in L_i$ and $\exists j \leq i$ such that $u \in S_j$)} {
        $L_i \gets L_i \setminus \{u\}\,$ (resp., $L_i \gets L_i \setminus \{v\}$)
    }
    \FKBRS{$u,v$}
}
\end{algorithm}

\smallskip
In the first phase, the algorithm iteratively adds vertices to the dominating set by recursively sampling a sequence of hitting sets. 
In each recursive call $i \geq 1$, we set a threshold $\gamma_i = \tfrac{|L_{i-1}|}{2k}-1$ and construct two sets $S_i$ and $L_i$.
The set $S_i$ is obtained by sampling each vertex of $L_{i-1}$ independently with probability $\min(c\ln(n) / \gamma_i, 1)$, for a sufficiently large constant $c$.
Roughly speaking the set $S_i$ is the hitting set of the vertices with degree more than $\gamma_i$ in $G[L_{i-1}]$.
Moreover, the set $S_i$ is w.h.p.\ small in size due to the sampling procedure.
The set $L_i$ is constructed as the subset of vertices of $L_{i-1}$ that do not belong to $S_i$ and do not have a neighbor in $S_i$. 
Given the property of the hitting set $S_i$, the set $L_i$ contains w.h.p.\ only vertices with degree at most $\gamma_i$ in $G[L_{i-1}]$.
The recursion starts with $L_0 = V$ and it ends when $|L_i| \le 4k$. 
%The fact that $L_i$ is small at the end of the recursion allows us to bound the size of the dominating set.

In the $i_{\text{th}}$ recursive call, if the size of $L_i$ is at most $|L_{i-1}|/2$ then a new recursive call begins. This implies that the depth of the recursion over all updates 
is bounded by $O(\log n)$.
On the other hand, if the size of the set $L_i$ is greater than $|L_{i-1}|/2$, the recursion pauses and the algorithm reports that there is an independent set in $G$ of size at least $k+1$. 
In this case $i$ may not be the final recursive call of the algorithm, because on future updates the algorithm can possibly continue the recursion.

Whenever an edge $(u,v)$ is inserted to $G$ during the first phase, we update the set $L_i$ by removing from it one of the endpoints if the other one is contained in $S_i$. Observe that edge insertions will eventually shrink the size of $L_i$, forcing the recursion 
to continue.

\smallskip
The second phase begins when the size of $L_i$ is at most $4k$, and at this moment the recursion ends.
We denote by $d$ the index of the last recursive call in the first phase, and let $S := \bigcup_{j=1}^{d} S_j \cup L_d$ be the union of the hitting sets of all recursive calls and of the set $L_d$. Notice that the set $S$ can be constructed explicitly during the first phase of the algorithm.
Also, in the updates following the second phase we never re-enter the first phase, and thus the set $S$ is not modified anymore.
We show in the analysis, that even though the set $S$ is random, it is always a dominating set in $G$.

At the beginning of the second phase, the dynamic \kbtors algorithm $\mathcal{B}$ of Theorem~\ref{th:fully_dyn_mis} is initialized on $G[S]$.
Whenever an edge $(u, v)$ is inserted to $G$ during
the second phase, the algorithm simply forwards the update to $\mathcal{B}$ if $u,v \in S$, and does nothing otherwise.

% \smallskip
% The second phase starts when the recursion stops because $|L_i| \leq 4k$. 
% Let us call $d$ the index of the last recursive call in the first phase. 
% Also, let $S := \bigcup_{i=1}^{d} S_i \cup L_d$ be the union of the hitting sets of all recursive calls together with the remaining vertices in $L_d$ from the last recursive call.
% We show in the analysis that $S$ is a dominating set on $G$.
% Note that in the updates following the second phase we never re-enter the first phase, and thus the set $S$ is not updated anymore.
% Also during the second phase, we run the dynamic \kbtors algorithm $\mathcal{B}$ of Theorem~\ref{th:fully_dyn_mis} on $G[S]$.

% \subparagraph{Edge insertions.}
% Consider now an edge insertion $(u, v)$ to $G$, and let $i$ be the current recursive call of the algorithm before the update arrives. 
% While the algorithm is in its first phase, we simply update the set $L_i$, and the algorithm continues with the recursion.
% Namely, if $u$ is in $L_i$ and $v$ is in $S_i$, then we remove $u$ from $L_i$.
% Similarly, if $v$ is in $L_i$ and $u$ is in $S_i$, then we remove $v$ from $L_i$.

% While the algorithm is in its second phase, the inserted edge is passed as an update to $\mathcal{B}$ running on $G[S]$ if  $u,v \in S$, and does nothing otherwise.

\newcommand{\new}[1]{\textcolor{blue}{#1}} 

\subsubsection{Analysis of the Algorithm}
The analysis consists of three claims.
First, we prove that whenever the algorithm reports that there is an independent set in $G$ of size at 
least $k + 1$, this is correct with high probability (w.h.p.).
Second, we show that there are $O(\log n)$ recursive calls and that w.h.p. the size of $S$ is $\tilde O(k)$.
Third, we prove that the set $S$ detected by the algorithm is indeed a dominating set in $G$.

\begin{lemma} \label{lem:kvert_nomis}
At any stage of the algorithm with $i \geq 1$, if $|L_i| > \frac{|L_{i-1}|}{2}$, then w.h.p.~there is an independent set in $G$ of size at least $k + 1$.
\end{lemma}
\begin{proof}
The threshold $\gamma_i$ is set to $\frac{|L_{i-1}|}{2k} - 1$,
and $S_i$ is obtained by sampling each vertex of $L_{i-1}$ independently with probability 
$\min(c\ln(n) / \gamma_i, 1)$, for a sufficiently large constant $c$.
Then by Lemma~\ref{lem:sampling}, it holds that w.h.p.~every vertex 
in $L_{i-1}$ of degree more than $\gamma_i$ in the induced subgraph $G[L_{i-1}]$ has a neighbor in $S_i$. 
Hence, w.h.p.~every vertex in $L_i$ is of degree at most $\gamma_i$ in $G[L_{i-1}]$.
As $G[L_i]$ is a subgraph of $G[L_{i-1}]$, w.h.p.~every vertex of $G[L_i]$ is of degree at most $\gamma_i$ in $G[L_i]$ as well.

Since w.h.p.~the maximum degree in $G[L_i]$ is bounded by $\gamma_i$, for any $T \subseteq L_i$ such that
$|T|= k$ (note that $|L_i|>4k$) \antonis{why is that necessary? Also maybe we should discard the case where $i = d$.} \antonis{add also ``ceil'' because of floating with division?}, it holds that w.h.p.~the number of vertices which are either in $T$ or have a neighbor in $T$ is
at most $k (\gamma_i+1) \le \frac{|L_{i-1}|}{2}$. 
By assumption we have that $|L_i| > \frac{|L_{i-1}|}{2}$,
and so $T$ cannot be a maximal independent set.
So it holds that w.h.p.~there is an independent set in $G[L_i]$ of size
at least $k + 1$. In turn, as $G[L_i]$ is an induced subgraph
of $G$, it holds that w.h.p.~there is an independent set in $G$ of size
at least $k + 1$ as well.
\end{proof}

\begin{lemma}\label{lem:dep_S_size}
    Over the sequence of updates, there are $d=O(\log n)$ recursive calls. Moreover, the size of $S$ is w.h.p.~$O(k \log^2 n)$.
\end{lemma}
\begin{proof}
    Regarding the first claim, at every recursive call $i \ge 1$, it holds that $|L_i| \le \frac{|L_{i-1}|}{2}$. Initially we have that $|L_0|=n$, and so, the depth of the recursion is $d=O(\log n)$.
    
    Regarding the second claim, at each recursive call $i \geq 1$, we sample each vertex of $L_{i-1}$ independently with probability $\min(c\ln(n) / \gamma_i, 1)$, for a sufficiently large constant $c$. 
    Recall that $\gamma_i = \frac{|L_{i-1}|}{2k} - 1$ and note that the sampling takes place only if $|L_{i-1}| > 4k$.
    Then
    \[
        \mathbb{E}[|S_i|] 
        \le |L_{i-1}| \cdot \frac{c \ln(n)}{\gamma_i}
        = |L_{i-1}| \cdot \frac{c \ln(n)}{\frac{|L_{i-1}|}{2k}-1}
        = |L_{i-1}| \cdot \frac{2k \cdot c\ln(n)}{|L_{i-1}| - 2k}
        = \frac{2k \cdot c\ln(n)}{1 - 2k/|L_{i-1}|}
        < 4k \cdot c \ln(n). 
    \]
    Moreover note that $|L_d|\le 4k$.
    Therefore, by linearity of expectation it holds that $\mathbb{E}[|S|] = |L_d| + \sum_{i=1}^d \mathbb{E}[|S_i|] = O(k \log^2 n)$.
    Finally, since $|S|$ is a sum of independent Poisson trials, a standard application of a Chernoff's bound implies that $|S|=O(k \log^2 n)$ with high probability.
\end{proof}
 
\begin{lemma}\label{lem:S_kern_G}
    The set $S$ is a dominating set in $G$.
\end{lemma}
\begin{proof}
    For a fixed vertex $v \in V \setminus S$, let $i$ be the minimum index such that $v \notin L_i$. 
    Note that such an index exists since $v \in V=L_0$ and so $i\ge 1$.
    If $v \notin L_i$, then by definition of $L_i$, vertex $v$ must have a neighbor in $S_i$.
    Therefore, every vertex $v \in V \setminus S$ has a neighbor in $S$ and the claim follows.
\end{proof}

Finally, Theorem~\ref{th:res_kbmis} which we restate for convenience, follows by the combination of Lemma~\ref{lem:kvert_nomis}, Lemma~\ref{lem:dep_S_size} and Lemma~\ref{lem:S_kern_G}.

\reskbmis*

Let us explain now why this algorithm actually solves w.h.p.~the incremental \kbttrs problem in $G$ as we argued before (see also Observation~\ref{obs:rs_beta_beta1}).
Recall that an independent set is a distance-$1$ independent set, and by the definition of the \kbttrs problem (see Definition~\ref{def:kbabrs}), we are allowed to report that
there is a distance-$1$ independent set of size at least $k + 1$. Hence, whenever this algorithm performs the operation stated in the first bullet of Theorem~\ref{th:res_kbmis}, the claim follows.

Regarding the operation stated in the second bullet of Theorem~\ref{th:res_kbmis}, let $\mathcal{B}$ be a \kbtors algorithm running on $G[S]$. 
If $\mathcal{B}$ reports that there is a distance-$1$ independent set in $G[S]$ of size at least $k + 1$,
then $\mathcal{B}$ correctly reports that there is a 
distance-$1$ independent set in $G$ of size at least $k + 1$ as well. This is because $G[S]$ is an induced subgraph of $G$,
and so, any distance-$1$ independent set in $G[S]$ is also a distance-$1$ independent set in $G$.

Otherwise, $\mathcal{B}$ returns a \tors $M$ of size at most $k$. 
Then for any vertex $v \in S$, we have that $v$ is of distance
at most $1$ from its closest vertex in $M$. Therefore, since every vertex $v \in V \setminus S$
has at least one neighbor in $S$ by Lemma~\ref{lem:S_kern_G}, we have that every vertex of $G$ is of distance at most $2$ from its closest vertex in $M$. Thus, the set $M$ is a \ttrs in $G$ of size at most $k$, and so the claim follows.

\subsection{Incremental $k$-Bounded $(2, 2)$-Ruling Set on $G_r$} \label{sec:incr_kbtt_Gr}
Our goal here is to extend Theorem~\ref{th:res_kbmis} to $r$-threshold graphs 
so that we can apply Lemma~\ref{lem:2b-approx_Gr} and maintain an incremental $k$-center solution. 
At a high level, our intention is to simulate the two phases of Algorithm~\ref{alg:2_k_MIS} on an $r$-threshold graph $G_r$.
Recall by Definition~\ref{def:threshold_graph} that for any pair of vertices $u,v \in V \times V$, there is an edge in $G_r$
if and only if the distance between $u$ and $v$ in $G$ is at most $r$.
The main challenges in the incremental setting are the following ones.
\begin{itemize}[topsep=0pt,itemsep=-1ex,partopsep=1ex,parsep=1ex]
    \item We cannot afford to explicitly maintain all the edges of $G_r$ in the incremental setting, because it is very expensive to run an incremental all-pairs shortest paths algorithm on $G$.
    
    \item A single edge insertion in the original graph $G$ could introduce multiple edge insertions in the $r$-threshold graph $G_r$.
\end{itemize}
Note that Algorithm~\ref{alg:2_k_MIS} does not need access to all edges of $G_r$ in order to process $G_r$.
Thus, our aim is to describe how to maintain all the necessary information 
that Algorithm~\ref{alg:2_k_MIS} needs, so as to run with implicit input the $r$-threshold graph $G_r$.

To extract the relevant information for the $r$-threshold graph $G_r$, we make use of the
incremental $(1+\epsilon)$-SSSP algorithm of Theorem \ref{th:incr_appr_sssp}. 
We note that using partially dynamic \emph{exact} SSSP algorithms for this step would be too slow for our purposes, as even in unweighted 
graphs we would require $\Omega(mr)$ time and $r$ could be very large (i.e., as big as $n$).
%The running times of exact incremental SSSP algorithms
%depend on the upper bound of the maintained distances, and in turn on the value of $r$. \antonis{should we add citations here?}
%Since the value of $r$ can be very big for our purpose (i.e., as big as $nW$),
%we maintain approximate instead of exact distances. Specifically, we use the approximate SSSP algorithm of
%Theorem~\ref{th:incr_appr_sssp}, whose running time is independent of the upper bound
%of the maintained distances.
Consequently, rather than explicitly maintaining $G_r$, we maintain an edge-subgraph $H$ of the 
$r'$-threshold graph $G_{r'}$, with $r':=(1+\epsilon)r$.
However, whenever the algorithm reports that there is an independent set in $H$ of size at least $k+1$,
we guarantee that this is also true for the $r$-threshold graph $G_r$.

We exploit the fact that Algorithm~\ref{alg:2_k_MIS} guarantees that the size of the dominating set is small (see Theorem~\ref{th:res_kbmis}). Hence, since during the second phase only the edges in the subgraph induced by the dominating set are needed, we argue 
based on Lemma~\ref{lem:dep_S_size} that during the whole second phase of the algorithm
we maintain $\tilde{O}(k)$ incremental $(1+\epsilon)$-approximate SSSP instances. Furthermore again by
Lemma~\ref{lem:dep_S_size}, we argue that during the whole first phase of the algorithm,
we maintain $\tilde{O}(1)$ incremental $(1+\epsilon)$-approximate SSSP instances. 
As a result, in total we maintain only $\tilde{O}(k)$ incremental $(1+\epsilon)$-approximate SSSP instances over the course of the algorithm,
and this is the main ingredient for the efficiency of the algorithm.
In particular, we prove the following theorem.

\begin{restatable}{theorem}{reskbmisGr}\label{thm:res_kbmisGr}
    Consider a weighted undirected graph $G = (V, E, w)$ subject to edge insertions, an integer $k \geq 1$, a positive parameter $r$ and a positive constant $\epsilon<1$. 
    Let $r' := (1+\epsilon)r$ and consider the threshold graphs $G_r$ and $G_{r'}$. There is a randomized algorithm which:
    \begin{itemize}[topsep=0pt,itemsep=-1ex,partopsep=1ex,parsep=1ex]
       \item either reports that there is an independent set in $G_r$ of size at least $k + 1$, and this is correct w.h.p.,
       \item or finds a dominating set $S \subseteq V$ of size $\tilde{O}(k)$ in an edge-subgraph $H$ of $G_{r'}$ and runs an incremental \kbtors algorithm
       $\mathcal{B}$ on $H[S]$ with the following condition: whenever $\mathcal{B}$ reports that there is an independent
       set in $H$ of size at least $k + 1$, then there is an independent set in $G_r$ of size at least $k + 1$.
    \end{itemize} 
    The total update time of the algorithm is w.h.p.\ $km^{1+o(1)}$.
\end{restatable}

Based on Observation~\ref{obs:rs_beta_beta1} and by the definitions of dominating set and \kbabrs problem (i.e., Definition~\ref{def:kern_set} and Definition~\ref{def:kbabrs}),
the next corollary immediately follows.

\begin{restatable}{corollary}{correskbmisGr}\label{cor:res_kbmisGr}
Consider the setting of Theorem~\ref{thm:res_kbmisGr}. There is a randomized algorithm which:
\begin{itemize}[topsep=0pt,itemsep=-1ex,partopsep=1ex,parsep=1ex]
   \item either reports that there is an independent set in $G_r$ of size at least $k + 1$, and this is correct w.h.p.,
   \item or runs an incremental \kbttrs algorithm $\mathcal{B}$ on an edge-subgraph $H$ of $G_{r'}$ with the following condition:
    whenever $\mathcal{B}$ reports that there is an independent
    set in $H$ of size at least $k + 1$, then there is an independent set in $G_r$ of size at least $k + 1$.
\end{itemize}
The total update time of the algorithm is w.h.p.\ $km^{1+o(1)}$.
\end{restatable}

\subsubsection{Overview of the Algorithm}
In the following, we describe the algorithm of Theorem~\ref{thm:res_kbmisGr}
which is an adaptation of Algorithm~\ref{alg:2_k_MIS} on approximate $r$-threshold graphs. Specifically, we adapt Algorithm~\ref{alg:2_k_MIS} to process an approximate
$r$-threshold graph $G_r$ implicitly. A pseudocode of the algorithm is provided in Algorithm~\ref{alg:incr_kbttrs_Gr}.

\begin{algorithm}[ht!]
\DontPrintSemicolon
\caption{\textsc{\kbttrs on $G_r$}}
\label{alg:incr_kbttrs_Gr}
\SetKwProg{Fn}{Procedure}{:}{\KwRet}
\SetKwFunction{FKBRS}{$k$-Bounded-Ruling-Set}
\SetKwFunction{FUpdate}{Insert}

\setcounter{AlgoLine}{0}

\tcp{In the preprocessing $i = 0$, $L_0 = V$, and $\FKBRS{}$ is called with no edge}

$r' \gets (1+\epsilon)r$

$H \gets (V, E_H)$ where $E_H \gets \emptyset$

\Fn{\FKBRS{$u,v$}}{
    \If(\tcp*[h]{first phase}){$|L_i| < 4k$}{
        \If{$i=0$ or $|L_i|\le \frac{|L_{i-1}|}{2}$}{
            $i \gets i+1$

            $\gamma_i \gets \frac{|L_{i-1}|}{2k} - 1$

            $S_i \gets$ sample vertices of $L_{i-1}$ independently with prob. $\min(10\ln(n) / \gamma_i, 1)$

            $S^{(i)} \gets \bigcup_{j=1}^{i} S_j$

            $\mathcal{A}_{S^{(i)}}.initialize(G, S^{(i)})$
            \tcp*{$\mathcal{A}_{S^{(i)}}$ is incremental approx. SSSP algorithm}
            \tcp*{$\mathcal{A}_{S^{(i)}}$ provides approx. distance $\delta_{S^{(i)}}(\cdot)$}

            $L_i \gets \{x \in V : \delta_{S^{(i)}}(x) > r'\}$

            $\FKBRS{u,v}$
        }
        \Else{
            \textbf{report} there exists an IS in $G_r$ of size at least $k+1$
        }
    }
    \Else(\tcp*[h]{second phase}){
        \If(\tcp*[h]{$\mathcal{B}$ is dynamic \kbtors algorithm (Theorem~\ref{th:fully_dyn_mis})}){$\mathcal{B}$ not initialized}{
            $d \gets i$
            
            $S \gets \bigcup_{j=1}^d S_j \cup L_d$
    
            \For{$s \in S$}{
                $\mathcal{A}_s.initialize(G,s)$
                \tcp*{$\mathcal{A}_s$ is incremental approx. SSSP algorithm}
                \tcp*{$\mathcal{A}_s$ provides approx. distance $\delta_s(\cdot)$}
            }

            $E_S \gets \{(u, v) \in S \times S : \delta_u(v) \le r'\}$ \tcp*{Edges of $H[S]$}
            
            $H \gets (V, E_H \cup E_S)$
            
            $\mathcal{B}.initialize(H[S])$
        }
        \Else{
            \For{$s \in S$}{
                $\mathcal{A}_s.insert(u,v)$
            }
            
            \While {$\exists\, a, b \in S$ such that $(a, b) \notin E_{S}$ and $\delta_a(b) \leq r'$} {
                $E_S \gets E_S \cup \{(a, b)\}$ \tcp*{$H[S] := (S, E_S)$}
                
                $\mathcal{B}.insert(H[S], a, b)$
            }
        }
    }
}

\vspace{1em}

\Fn{\FUpdate{$u, v$}}{
    $G \gets (V,\, E \cup \{u,v\})$
    
    $\mathcal{A}_{S^{(i)}}.insert(u,v)$
    
    \While {$\exists\, x \in L_i$ such that $\delta_{S^{(i)}}(x) \leq r'$} {
        $L_i \gets L_i \setminus \{x\}$
    }

    $\FKBRS{u,v}$
}
\end{algorithm}
\antonis{why $L_d$ is missing in second phase in the construction of $S$?}
\antonis{Update $H$ and $E_S$ in the pseudocode?}
\antonis{Replaced $\exists$ in line 31 of Insert()}
For a fixed value of $r$ and $\epsilon$, let $r' := (1+\epsilon)r$. Also let $H$ be an initially empty graph with vertex set $V$. Consider a recursive call $i \geq 1$ of Algorithm~\ref{alg:2_k_MIS} during the first phase. 
The sampling step for obtaining the set $S_i$ does not need access to the edges of the input graph,
but only to the vertices of the input graph. Thus, each hitting set $S_i$ can be explicitly constructed. In turn, 
the union of the sampled sets $S^{(i)} = S_1 \cup \cdots \cup S_i$ is explicitly constructed as well.

The next step of Algorithm~\ref{alg:2_k_MIS} is to compute the size of $L_i$, 
and decide how to proceed with the recursion depending on the sizes of $L_{i-1}$ and $L_i$.
A simulation of Algorithm~\ref{alg:2_k_MIS} on $G_r$ would construct the set $L_i$ as the set
of vertices which are of distance more than $r$ in $G$ from their closest vertex in $S_i$. 
Nevertheless, as we use an approximate SSSP algorithm, we construct the set $L_i$ in a slightly
different way as follows.
In the beginning of the recursive call, we set $S = S^{(i)}$ and $L_i = L_{i-1} \setminus S_i$.
At this point, we maintain the incremental $(1+\epsilon)$-SSSP algorithm of Theorem~\ref{th:incr_appr_sssp}
with super-source $S$ on $G$, providing distance estimates $\delta_S(\cdot)$.\footnote{Namely, we introduce a fake root $x$ and 
add an edge $(x, v)$ of zero weight, for every $v \in S$. Then, we run the approximate SSSP algorithm
with source $x$ on $G$.} 
Whenever the distance estimate $\delta_S(v)$ of a vertex $v \in V$ becomes smaller than $(1+\epsilon)r$,
we remove $v$ from $L_i$ and add the edge $(a, v)$ to $H$, where $a$ is the corresponding vertex of $S$ for the distance estimate $\delta_S(v)$.
Therefore, we continue with the recursion as in Algorithm~\ref{alg:2_k_MIS}
by constructing the set $L_i$ in this way, and in turn computing its size. Moreover, since at every new recursive call the set $S^{(i)}$ is modified (i.e., we may sample more vertices), at every new recursive call we set $S = S^{(i)}$
and we restart the incremental 
$(1+\epsilon)$-SSSP algorithm with super-source $S$ on $G$.

\smallskip
As in Algorithm~\ref{alg:2_k_MIS}, the second phase begins when the size of $L_i$ is at most $4k$, and at this moment the recursion ends.
We denote by $d$ the index of the last recursive call in the first phase,
and $S$ is updated to $S := S^{(d)} \cup L_d$.
During the second phase, Algorithm~\ref{alg:2_k_MIS} has to maintain an incremental \kbtors 
algorithm on $G_r[S]$. Instead, we maintain an incremental \kbtors algorithm on a subgraph $H[S]$ of $G_{r'}[S]$. The subgraph $H[S]$ is maintained explicitly, as follows.
Let $E_S$ be the initially empty set consisting of the edges in $H[S]$.
For each vertex $v \in S$, we maintain the incremental $(1+\epsilon)$-SSSP algorithm of 
Theorem~\ref{th:incr_appr_sssp} with source $v$ on $G$,
providing distance estimates $\delta_v(\cdot)$. Then for any two vertices $u, v \in S$, 
whenever we have that $\delta_u(v) \leq (1+\epsilon)r$ or $\delta_v(u) \leq (1+\epsilon)r$,
the edge $(u, v)$ is added to $E_S$. \antonis{In pseudocode we only check $\delta_u(v)$, but here also $\delta_v(u)$. Are the two descriptions equivalent?}

Thus during the second phase, we maintain the incremental \kbtors algorithm $\mathcal{B}$ of Theorem~\ref{th:fully_dyn_mis} on
$(S, E_S)$. Whenever $\mathcal{B}$ reports that there is an independent set in $(S, E_S)$ of size at least $k + 1$,
we report that there is an independent set in $G_r$ of size at least $k + 1$.
In the analysis we argue that $S$ is a dominating set in $H$. Hence, this implies
that $\mathcal{B}$ solves the incremental \kbttrs problem in $H$.

\paragraph{Edge insertions.}
Consider an edge insertion to $G$, and let $i$ be the current recursive call of the algorithm before 
the update arrives. While the algorithm is in the first phase, the update is passed
to the incremental $(1+\epsilon)$-SSSP algorithm with super-source $S^{(i)}$, and the corresponding
set $L_i$ is updated accordingly. If the algorithm is in (or enters) the second phase after an edge insertion, 
then for every vertex $v \in S$, the inserted edge is passed as an update to the incremental 
$(1+\epsilon)$-SSSP algorithm with source $v$.
Notice that an edge insertion in the original graph $G$ could introduce multiple updates to $E_S$. Thus whenever an edge $(a, b)$ is added to $E_S$, the edge $(a, b)$ is passed as an update to the incremental \kbtors algorithm $\mathcal{B}$ running on $(S, E_S)$.

\subsection{Analysis of the Algorithm}
Our goal here is to prove Theorem~\ref{thm:res_kbmisGr}. 
Note that if we had access to exact distances and we removed a vertex $v$ from $L_i$ whenever its distance estimate is at most $r$, then the correctness would follow from the arguments of the previous section. 
However, since for efficiency purposes we are utilizing approximate distances, the analysis has to be adapted.
In our case we remove a vertex $v$ from $L_i$ whenever its approximate distance estimate is at most $r'=(1+\epsilon)r$.
The next lemma is similar to Lemma~\ref{lem:kvert_nomis} but now applied to the $r$-threshold graph $G_r$.

\begin{lemma} \label{lem:kvert_nomis_Gr}
    At any stage of the algorithm with $i \geq 1$, if $|L_i| > \frac{|L_{i-1}|}{2}$, then w.h.p.\ there is an independent in $G_r$ of size at least $k + 1$.
\end{lemma}
\begin{proof}
    The threshold $\gamma_i$ is set to $\frac{|L_{i-1}|}{2k} - 1$,
    and $S_i$ is obtained by sampling each vertex of $L_{i-1}$ independently  
    with probability $\min(c\ln(n) / \gamma_i, 1)$, for a sufficiently large constant $c$.
    Then by Lemma~\ref{lem:sampling}, it holds that w.h.p.\ every vertex $v$
    in $L_{i-1}$ of degree more than $\gamma_i$ in the induced subgraph $G_r[L_{i-1}]$ has a neighbor in $S_i$. 
    This is equivalent of saying that w.h.p.~every vertex $v$ of degree more than $\gamma_i$ in $G_r[L_{i-1}]$ is within distance $r$ from a vertex of $S_i$ in $G$, that is, $d_G(v, S_i) \leq r$. Then, by Theorem~\ref{th:incr_appr_sssp}
    we have that $\delta_S(v) \leq (1 + \epsilon)r$, which means that $v$ has been removed from $L_i$.
    In turn, this implies that w.h.p.~every vertex in $L_i$ is of
    degree at most $\gamma_i$ in $G_r[L_{i-1}]$. 
    As $G_r[L_i]$ is a subgraph of $G_r[L_{i-1}]$, w.h.p.~every vertex of $G_r[L_i]$ is of degree at most $\gamma_i$ in $G_r[L_i]$ as well.
    
    Since w.h.p.~the maximum degree in $G_r[L_i]$ is bounded by $\gamma_i$, the claim follows by applying the same process of the second paragraph of Lemma~\ref{lem:kvert_nomis} on $G_r[L_i]$.
\end{proof}

Notice that after computing the size of $L_i$, the recursion continues in the same way as in Algorithm~\ref{alg:2_k_MIS}. The next lemma says that Lemma~\ref{lem:dep_S_size} holds in this algorithm as well. Recall that $d$ is the recursive call after the first phase has ended and just before the second phase begins (i.e., $d$ is the final depth of the recursion). 

% \antonis{We just repeat the same lemma. Maybe we should say it differently.}
\begin{lemma} \label{lem:dep_S_size_Gr}
    Over the sequence of updates, there are $d=O(\log n)$ recursive calls. Moreover, the size of $S$ is w.h.p.~$O(k \log^2 n)$.
\end{lemma}

Remember that we want to use Algorithm~\ref{alg:incr_kbttrs_Gr} as a subroutine in the
incremental $k$-center algorithm. By using the next property of $H$, we argue that only an extra $(1 + \epsilon)$ factor 
shows up in the approximation ratio of the $k$-center algorithm.

\begin{lemma}
    The graph $H$ is a subgraph of the $r'$-threshold graph $G_{r'}$, where $r'= (1+\epsilon)r$. 
\end{lemma}
\begin{proof}
    Let $(u, v) \in E(H)$ be an edge of the graph $H$. Assume that the edge has been 
    added during the first phase of the algorithm. Then, w.l.o.g. it must be the case that 
    $u \in S$ and $\delta_S(v) \leq r'$.
    Based on Theorem~\ref{th:incr_appr_sssp}, 
    % \antonis{should we add in theorem: distance estimates do not underestimate?} 
    the distance estimate $\delta_S(\cdot)$ does not underestimate the distances, and so
    we have that $d_G(u, v) \leq \delta_S(v) \leq r'$. 
    Thus by definition, the edge $(u, v)$ is part of $G_{r'}$ as well.

    Similarly, assume that the edge has been added during the second phase of the algorithm. Then, w.l.o.g. it must be the case that $u, v \in S$ and $\delta_u(v) \leq r'$.
    Using a similar argument as before, we conclude that every edge of $H$ is part of $G_{r'}$.
\end{proof}

During the second phase, whenever the incremental \kbtors algorithm $\mathcal{B}$ reports that there is an
independent set in $H$ of size at least $k+1$, we report that there is an independent set in $G_r$ 
of size at least $k + 1$. Since algorithm $\mathcal{B}$ is running on $H[S]$, the following
lemma states that in this case, there is definitely an independent set in $G_r$
(and not just with high probability) of size at least $k + 1$.

\begin{lemma}
    Any independent set in $H[S]$ is also an independent set in $G_r[S]$.
\end{lemma}
\begin{proof}
    Let $M$ be an independent set in $H[S]$, and suppose to the contrary that $M$ is not an independent
    set in $G_r[S]$. Then there must exist two vertices $u, v \in S \cap M$, such that 
    the edge $(u, v)$ belongs to $G_r[S]$ but not to $H[S]$. Since $(u, v)$ belongs to $G_r$, the distance between $u$ and $v$ in $G$ is at most $r$ (i.e., $d_G(u, v) \leq r$). Also as $u \in S$, in the algorithm
    we maintain the incremental $(1+\epsilon)$-SSSP algorithm with source $u$ on $G$, and by Theorem~\ref{th:incr_appr_sssp} it holds that $\delta_u(v) \leq (1+\epsilon)d_G(u, v) \leq (1+\epsilon)r$.
    Hence as $v \in S$, the algorithm must have added the edge $(u, v)$ to $H$, which contradicts
    the assumption that the edge $(u, v)$ does not belong to $H[S]$.
\end{proof}

\begin{lemma} \label{lem:kern_set_H}
    The set $S$ is a dominating set in $H$.
\end{lemma}
\begin{proof}
    The claim follows by applying the 
    proof of Lemma~\ref{lem:S_kern_G} on $H$.
\end{proof}

\paragraph{Running time.}
During the first phase, the incremental $(1+\epsilon)$-SSSP algorithm with super-source $S^{(i)}$ on $G$, is restarted as many times
as the number of the recursive calls. By Lemma~\ref{lem:dep_S_size_Gr}, there are at most $O(\log n)$ recursive calls in total,
and by Theorem~\ref{th:incr_appr_sssp}, the total update time of the incremental $(1+\epsilon)$-SSSP algorithm is $m^{1 + o(1)}$. 
Thus, the total update time charged for the first phase of the algorithm is $m^{1 + o(1)}$.

During the second phase, for every vertex $v \in S$, we maintain the incremental 
$(1+\epsilon)$-SSSP algorithm of Theorem~\ref{th:incr_appr_sssp} with source $v$ on $G$.
By Lemma~\ref{lem:dep_S_size_Gr}, the size of $S$ is w.h.p.\ $\tilde{O}(k)$, and so the total
update time for maintaining the edge set $E_S$ is $km^{1 + o(1)}$.

\begin{observation} \label{obs:E_S_nondec}
    Since $G$ is subject to edge insertions, the edge set $E_S$ of $H[S]$ is non-descreasing.
\end{observation}

The \kbtors algorithm $\mathcal{B}$ of Theorem~\ref{th:fully_dyn_mis} is running on 
$(S, E_S)$ (i.e., the induced subgraph $H[S]$). Also, 
as the edge set $E_S$ contains only edges between vertices in $S$, the maximum size of $E_S$ is w.h.p.\
$\tilde{O}(k^2)$.
Then based on Theorem~\ref{th:fully_dyn_mis} and Observation~\ref{obs:E_S_nondec}, 
the total update time charged for $\mathcal{B}$ is $\tilde{O}(k^2)$, which 
(when amortized over the $\Omega(k^2)$ total edge insertions to $E_s$) amounts to 
an amortized update time of $\tilde{O}(1)$.
This concludes the running time analysis of Theorem \ref{thm:res_kbmisGr}.

\vspace{1em}
Finally, Theorem~\ref{thm:res_kbmisGr} follows by combining all the previous lemmas. In turn, Corollary~\ref{cor:res_kbmisGr} which we will use follows, and we restate it here for convenience.

\correskbmisGr*

\subsection{Incremental $k$-Center on Graphs: Putting It Together}
At this point, we have developed all the necessary tools in order to obtain our main theorem 
for the incremental $k$-center problem on graphs. The idea is to combine the reduction of Lemma~\ref{lem:2b_Gr_H} with Algorithm~\ref{alg:incr_kbttrs_Gr} of Theorem~\ref{thm:res_kbmisGr}.

\incrkcentfappr*
\begin{proof}
    Observe that algorithm $\mathcal{A}$ inside Lemma~\ref{lem:2b_Gr_H}
    with $\beta = 2$, has the same properties of the algorithm
    in Corollary~\ref{cor:res_kbmisGr}. Hence,
    let $\mathcal{A}$ be the Algorithm~\ref{alg:incr_kbttrs_Gr} of Corollary~\ref{cor:res_kbmisGr}, and
    $\epsilon_1 = \frac{\epsilon}{12}$. Based on Lemma~\ref{lem:2b_Gr_H},
    by running $\mathcal{A}$ with input $G, r, \epsilon_1$, for each
    $r \in \{(1 + \epsilon_1)^i \mid (1 + \epsilon_1)^i \leq nW, i \in \mathbb{N}\}$, we get an
    incremental $4(1+\epsilon_1)(1+\epsilon_1)$-approximation algorithm for the
    $k$-center problem. As $\epsilon < 1$ and $\epsilon_1 = \frac{\epsilon}{12}$, the approximation
    ratio is $(4 + \epsilon)$.

    Regarding the running time, by Corollary~\ref{cor:res_kbmisGr} the total update time of $\mathcal{A}$ is w.h.p.~$km^{1+o(1)}$. Since we run $\mathcal{A}$ for at most $O(\log_{1+\epsilon_1}(nW))$ different values of $r$, the total update time
    of the algorithm remains $km^{1+o(1)}$.
\end{proof}

\section{Decremental $k$-Center on Graphs}\label{sec:decremental}
In the decremental setting, the input graph of the $k$-center instance is subject to edge deletions.
Based on Lemma~\ref{lem:2b-approx_Gr}, in order to get a $(2 + \epsilon)$-approximation decremental
algorithm for the $k$-center problem, it is sufficient to develop a decremental algorithm for the 
\kbtors problem on $r$-threshold graphs. To maintain the necessary information for the $r$-threshold
graphs, we use a decremental SSSP algorithm on $G$.

\subsection{Decremental $k$-Bounded $(2, 1)$-Ruling Set on $G_r$}
For the sake of efficiency, in order to maintain the necessary information
for the $r$-threshold graphs, we make use of the approximate SSSP algorithm of 
Theorem~\ref{th:decr_appr_sssp}. Thus, we obtain instead the following theorem
which is a slight relaxation of the decremental \kbtors problem on $r$-threshold graphs.
This is still sufficient for the $k$-center problem, as Lemma~\ref{lem:2b_Gr_H}
suggests.

\begin{theorem} \label{th:res_dec_tors}
    Consider a weighted undirected graph $G = (V, E, w)$ subject to edge deletions, an integer $k \geq 1$, a positive parameter $r$ and a positive constant $\epsilon < 1$. 
    Let $r' := (1+\epsilon)r$ and consider the threshold graphs $G_r$ and $G_{r'}$.
    There is a deterministic algorithm which:
    \begin{itemize}[topsep=0pt,itemsep=-1ex,partopsep=1ex,parsep=1ex]
        \item either reports that there is an independent set in $G_r$ of size at least $k + 1$,
        \item or runs a decremental \kbtors algorithm $\mathcal{B}$ on an edge-subgraph $H$ of $G_{r'}$ with
        the following condition: whenever $\mathcal{B}$ reports that
        there is an independent set in $H$ of size at least $k + 1$, then
        there is an independent set in $G_r$ of size at least $k + 1$.
    \end{itemize}
    The total update time of the algorithm is $km^{1+o(1)}$.
\end{theorem}

Recall that in the definition of an \abrs, the first property is that the distance between any two vertices
in the \abrs is at least $\alpha$. The crucial observation here is that under edge deletions, the distance between any 
two vertices is non-decreasing. Hence the first property is preserved in the decremental setting, and this is the major
ingredient for the algorithm.

\subsubsection{Overview of the Algorithm}
For a fixed value of $r$ and $\epsilon$, let $r' := (1 + \epsilon)r$.
In the beginning of the algorithm of Theorem~\ref{th:res_dec_tors}, we execute a static \kbtors algorithm $\mathcal{B}$ on $G_r$.
One simple algorithm for this problem is to run $k$ times the Dijkstra's algorithm on $G$.
In particular, at each iteration we choose a vertex $s$ which has not been covered yet, 
and we run Dijkstra's algorithm on $G$ with source $s$. Then, 
every vertex $v$ of distance at most $r$ from $s$ is set as covered, and the same process is repeated at most $k$
times. The running time of this algorithm is clearly $\tilde{O}(mk)$.

%For the problem of finding an MIS on $r$-threshold graphs, there are also
%faster algorithms~\cite{abboud2023fine, thorup2005quick} with $\tilde{O}(m)$ expected time. 

Assume that algorithm $\mathcal{B}$ returns a \tors $M$ in $G_r$ of size at most $k$.
Next, we initialize a decremental approximate SSSP algorithm $\mathcal{A}$ with super-source $M$ on $G$, providing distance 
estimates $\delta(\cdot)$.\footnote{Namely, we introduce a fake root $x$ and 
add an edge $(x, v)$ of zero weight, for every $v \in M$. Then, we run a decremental approximate SSSP algorithm
with source $x$ on $G$.}
Specifically, we use the $(1 + \epsilon)$-approximate SSSP algorithm of Theorem~\ref{th:decr_appr_sssp}.
Also let $H$ be a graph whose edge set  contains all the edges $(u, v) \in V \times V$ such that
$\delta(v) \leq r'$ and $u \in M$ is the corresponding vertex for the distance estimate $\delta(v)$.
%\antonis{explain how we find $u$.}
The graph $H$ can be explicitly constructed during the previous step.

Whenever there is an edge deletion in $G$, we pass this update to $\mathcal{A}$. In turn,
this update can possibly increase the distance estimate $\delta(\cdot)$ of some vertices. 
In particular, whenever the distance estimate $\delta(v)$ of a vertex $v \in V$ becomes greater
than $r'$, we add $v$ to $M$, and the algorithm $\mathcal{A}$ is restarted with super-source the
modified set $M$. Moreover, the graph $H$ is recomputed from scratch as before.

At any moment, if the size of $M$ has exceeded $k$, the algorithm reports that there is an independent set
in $G_r$ of size at least $k + 1$, and we do not restart the algorithm $\mathcal{A}$ anymore.

\subsubsection{Analysis of the Algorithm}

Our goal here is to prove Theorem~\ref{th:res_dec_tors}.
Initially the static algorithm produces a \tors $M$ in $G_r$.
At any moment, if the size of $M$ becomes at least $k + 1$, the algorithm reports
that there is an independent set in $G_r$ of size at least $k + 1$.
The next lemma shows the correctness of this step.

\begin{lemma} \label{lem:dist_alpha_IS}
    If the size of $M$ is at least $k + 1$, then there is
    an independent set in $G_r$ of size at least $k + 1$.
\end{lemma}
\begin{proof}
    Initially the set $M$ is a \tors in $G_r$, and by definition $M$ is 
    also an independent set in $G_r$. Thus, if the size of $M$
    is at least $k + 1$ after the execution of the static algorithm, 
    the set $M$ remains an independent set in $G_r$ under edge deletions, and the claim holds.

    Hence, we can assume that the size of $M$ became at least $k + 1$ after some edge deletions.
    We prove the claim by contradiction. 
    Suppose to the contrary that $M$ is not an independent set in $G_r$ after an edge deletion. 
    In this case, the algorithm must have added 
    a vertex $v$ to $M$ which has a neighbor $u \in M$ in $G_r$ (i.e., $d_G(u, v) \leq r$). 
    Since $u \in M$, in the algorithm
    we maintain the decremental $(1+\epsilon)$-SSSP algorithm with super-source $M$, and by 
    Theorem~\ref{th:decr_appr_sssp} it holds that
    $\delta(v) \leq (1+\epsilon)d_G(u, v) \leq r'$.
    But then, the algorithm does not add $v$ to $M$ which yields a contradiction.
\end{proof}

Assume that the size of the solution $M$ is at most $k$. The next lemma shows that
the algorithm maintains a decremental \kbtors algorithm on an edge-subgraph $H$ of $G_{r'}$. 

\begin{lemma} \label{lem:dec_M_tors}
    The graph $H$ is a subgraph of $G_{r'}$. 
    Moreover, if $|M| \leq k$, then the set $M$ is always a \tors in $H$.
\end{lemma}
\begin{proof}
    Let $(u, v) \in E(H)$ be an edge in $H$. 
    Then, w.l.o.g. it must be the case that $u \in M$ and $\delta(v) \leq r'$.
    Based on Theorem~\ref{th:decr_appr_sssp}, the distance estimate $\delta(\cdot)$ does not underestimate the distances,
    and so we have that $d_G(u, v) \leq \delta(v) \leq r'$. 
    Thus by definition, the edge $(u, v)$ is part of $G_{r'}$ as well.
    
    The algorithm adds the vertex $v$ to $M$ only if the distance estimate $\delta(v)$ becomes
    greater than $r'$, while the edge $(u, v)$ is part of $H$ only if $\delta(v)$ is at most $r'$. This implies that the set $M$ is an independent set in $H$.
    Furthermore, whenever the distance estimate $\delta(v)$ of a vertex $v \in V \setminus M$ becomes
    greater than $r'$, the set $M$ and the graph $H$ are recomputed. 
    This implies that the distance estimate of any vertex $v \in V \setminus M$ is at most $r'$.
    By construction of $H$, there must exist an edge $(u, v)$ in $H$, where $u \in M$ is the corresponding vertex of $\delta(v)$.
    Hence, we can conclude that the set $M$ is a \tors in $H$.
\end{proof}

\subparagraph{Running time.} The running time of the simple static algorithm is $\tilde{O}(mk)$.
By Theorem~\ref{th:decr_appr_sssp}, the total time of the decremental approximate SSSP algorithm is $m^{1+o(1)}$.
As the decremental approximate SSSP algorithm is restarted at most $k$ times,
the total update time of the algorithm is $km^{1+o(1)}$. Notice that the time to detect whether
a distance estimate is greater than $r' = (1+\epsilon) r$ is incorporated in the update time of the decremental approximate SSSP algorithm. 

\vspace{1em}

Finally, Theorem~\ref{th:res_dec_tors} follows by combining 
Lemma~\ref{lem:dist_alpha_IS} and Lemma~\ref{lem:dec_M_tors} together with the analysis of the running time.

\subsection{Decremental $k$-Center on Graphs: Putting It Together}
We combine Theorem~\ref{th:res_dec_tors} with Lemma~\ref{lem:2b_Gr_H}
to obtain the next theorem for the decremental $k$-center problem on graphs.
A pseudocode of the algorithm of Theorem~\ref{th:dec_kcent_2appr} is provided in Algorithm~\ref{alg:dec_kcenter}. 

\begin{algorithm}[ht!]
\DontPrintSemicolon
\caption{\textsc{decremental $(2 + \epsilon)$-approximation algorithm for $k$-center}{}}
\label{alg:dec_kcenter}
\SetKwFunction{FMaximalDistrIS}{MaximalDistrIS}
\SetKwProg{Fn}{Function}{:}{\KwRet}

\setcounter{AlgoLine}{0}
\SetAlgoLined

\Fn{\FMaximalDistrIS{}} {
    $i \gets 0$
    
    \While {$G \setminus \bigcup_{j=1}^i C_j \neq \emptyset$ and $i \leq k$} {
        $u \gets $ arbitrary vertex from $G \setminus \bigcup_{j=1}^i C_j$
    
        $i \gets i+1$
    
        $c_i \gets u$

        $C_i \gets$ cluster with center $c_i$ and radius $r$
    }

    \KwRet i
}

\vspace{1em}

\SetKwProg{Fn}{Procedure}{:}{\KwRet}
\SetKwFunction{FFindRadius}{FindRadius}
\Fn{\FFindRadius{}} {
    \While {\FMaximalDistrIS{} > k} {
        $r \gets (1 + \epsilon) \cdot r$
    }
}

\vspace{1em}

\SetKwProg{Fn}{Procedure}{:}{\KwRet}
\SetKwFunction{FPreprocessing}{Preprocessing}
\Fn{\FPreprocessing{}} {
    $r \gets 1$
    
    \FFindRadius{}

    $M \gets \{c_1, \dots, c_i\}$
    
    $\mathcal{A}.initialize(G, M)$
    \tcp*{$\mathcal{A}$ is decremental approx. SSSP algorithm with distance estimates $\delta(\cdot)$}
}

\vspace{1em}

\SetKwProg{Fn}{Procedure}{:}{\KwRet}
\SetKwFunction{FUpdate}{Update}
\Fn{\FUpdate{$u$, $v$}} {
    $G \gets (V, E \setminus (u,v))$

    $\mathcal{A}.delete(u,v)$

    \While{$\exists x \in V$ such that $\delta(x) > (1+\epsilon) r$} {
        \If{$i < k$} {
            $i \gets i + 1$

            $c_i \gets x$

            $M \gets \{c_1, \dots, c_i\}$
            
            $C_i \gets$ cluster with center $c_i$ and radius $r$
        }
        \Else {
            $r \gets (1 + \epsilon) \cdot r$
            
            \FFindRadius{}
        }

        $\mathcal{A}.restart(G, M)$
    }
}
\end{algorithm}

\deckcenttappr*
\begin{proof}
    Observe that algorithm $\mathcal{A}$ inside Lemma~\ref{lem:2b_Gr_H} with $\beta = 1$, 
    has the same properties of the algorithm in Theorem~\ref{th:res_dec_tors}. Hence, 
    let $\mathcal{A}$ be the algorithm of Theorem~\ref{th:res_dec_tors}, and 
    $\epsilon_1 = \frac{\epsilon}{6}$. Based on Lemma~\ref{lem:2b_Gr_H},
    by running $\mathcal{A}$ with input $G, r, \epsilon_1$, for each
    $r \in \{(1 + \epsilon_1)^i \mid (1 + \epsilon_1)^i \leq nW, i \in \mathbb{N}\}$, we get a
    deterministic decremental $2(1+\epsilon_1)(1+\epsilon_1)$-approximation algorithm for the
    $k$-center problem. As $\epsilon < 1$ and $\epsilon_1 = \frac{\epsilon}{6}$, the approximation
    ratio is $(2 + \epsilon)$.

    Regarding the running time, by Theorem~\ref{th:res_dec_tors} the total update time of $\mathcal{A}$ is $km^{1+o(1)}$. Since we run $\mathcal{A}$ for at most $O(\log_{1+\epsilon_1}(nW))$ different values of $r$, the total update time
    of the algorithm remains $km^{1+o(1)}$. 
\end{proof}

\section{Fully Dynamic $k$-Center on Graphs}\label{sec:fullydynamic}
In this section we describe how to maintain a $(2+\epsilon)$-approximate solution to the $k$-center problem on fully dynamic graphs.
We start by reviewing Gonzalez's algorithm, a classical $2$-approximation algorithm to the $k$-center problem in the static setting. Afterwards, we describe how to adapt it to the fully dynamic setting by using fully dynamic approximate SSSP algorithms.

\subsection{Gonzalez's Algorithm}
Gonzalez's algorithm~\cite{Gonzalez85} is a well-known greedy algorithm for the $k$-center problem on (possibly weighted and directed) graphs\footnote{The algorithm is also used for the $k$-center problem in metric spaces.}. 
It works as follows:
\begin{enumerate}[topsep=0pt,itemsep=-1ex,partopsep=1ex,parsep=1ex]
    \item pick as first center an arbitrary vertex $c_1 \in V$ and set $C = \{c_1\}$;
    \item while $|C| < k$, pick the next center $c_i \in \argmax_{v \in V}\, d_G(C, v)$ and set $C = C \cup \{c_i\}$; 
    \item return the set of centers $C$.
\end{enumerate}
\begin{theorem}
Gonzalez's algorithm computes a 2-approximation for the $k$-center problem on graphs and a standard implementation runs in time $O(k(m + n \log n))$.
\end{theorem}
\begin{Definition}[$\alpha$-approximate Gonzalez's algorithm]
For $\alpha \ge 1$, an $\alpha$-approximate Gonzalez's algorithm is a relaxation of Gonzalez's algorithm that picks the next center $c_i$ in step 2 above such that $d_G(C, c_i) \ge \alpha^{-1}\cdot \max_{v \in V} d_G(C,v)$.
\end{Definition}
\begin{theorem}[{\cite[Lemma~4.1]{AbboudCLM23}}]\label{thm:alpha gonzalez}
For $\alpha \ge 1$, an $\alpha$-approximate Gonzalez's algorithm computes a $2\alpha$-approximation for the $k$-center problem on graphs.
\end{theorem}

\subsection{Fully Dynamic $k$-Center via Fully Dynamic $(1+\epsilon)$-SSSP}

Assuming that we have a fully dynamic $(1+\epsilon)$-SSSP data structure,
we show how to use this to get a fully dynamic $k$-center data structure in Algorithm~\ref{alg:fullydyn_kcenter}.

\begin{algorithm}[ht!]
\caption{\textsc{Fully dynamic $2(1+\epsilon)$-approximation k-center}}
\label{alg:fullydyn_kcenter}
\DontPrintSemicolon
\setcounter{AlgoLine}{0}
\SetAlgoLined

\SetKwInOut{Input}{Input}
\SetKwInOut{Output}{Output}

% \Input{~Unweighted undirected graph $G=(V,E)$; integer $k\ge 1$}
% \Output{~Set of $k$ centers $C$}

\SetKwFunction{FSimulateGonzalez}{SimulateGonzalez}
\SetKwProg{Fn}{Function}{:}{\KwRet}
\Fn{\FSimulateGonzalez{$\mathcal{D}$, s, k}} {
    $C = \emptyset$
    
    \For{$i=1,...,k$}{
        $c_i \gets x \in \argmax_{v\in V} \delta_{G'}(s, v)$
        \tcp*{$c_1 \gets \text{arbitrary } v \in V$}
        
        \tcp*{$s$ is disconnected at $i=1$; $\delta_{G'}(s, v)=\infty,\, \forall v \in V$; ties broken arbitrarily}
        
        $\mathcal{D}.insert(s, c_i)$

        $C \gets C \cup \{c_i\}$
    }
    
    \For{$i=1,...,k$}{        
        $\mathcal{D}.delete(s, c_i)$
    }
    
    \KwRet $C$
}
\vspace{1em}
\SetKwFunction{FPreprocessing}{Preprocessing}
\Fn{\FPreprocessing{G, k}} {
    $G' \gets (V \cup \{s\},\, E)$ \tcp*{augment $G$ with super source $s$}
    
    $\mathcal{D} \gets \textsc{Initialize}(G', s)$ \tcp*{$\mathcal{D}$ is fully dynamic $(1+\epsilon)$-SSSP, with approx. distance $\delta_{G'}(s, v)$}
    
    $C \gets$ \FSimulateGonzalez{$\mathcal{D}$, s, k}
}
\vspace{1em}
\SetKwFunction{FUpdate}{Update}
\Fn{\FUpdate{$u$, $v$}} {
    $\mathcal{D}.update(u,v)$ \tcp*{either insert or delete edge $(u,v)$}
    
    $C \gets$ \FSimulateGonzalez{$\mathcal{D}$, s, k}
}
\end{algorithm}

\begin{theorem}\label{thm: fully dynamic black box}
Given a graph $G=(V,E)$, a positive parameter $\epsilon\le 1/2$, and a fully dynamic data structure that maintains $(1+\epsilon)$-approximate distances from a single source $s \in V$ with worst-case update time $T(n, m, \epsilon)$, Algorithm~\ref{alg:fullydyn_kcenter} maintains a $2(1+4\epsilon)$-approximate solution to fully dynamic $k$-center in time $O(k \cdot ( T(n,m, \epsilon)+n))$. 
\end{theorem}
\begin{proof}
We prove that the procedure \textsc{SimulateGonzalez} in Algorithm~\ref{alg:fullydyn_kcenter} is a $(1+4\epsilon)$-approximate Gonzalez's algorithm, hence the claim about the approximation follows by Theorem~\ref{thm:alpha gonzalez}.

Note that the procedure runs on $G'$, which is a copy of $G$ with an additional super-source vertex $s$ which is initially disconnected. 
Let us call $\mathcal{D}$ the data structure used to maintain the $(1+\epsilon)$-approximate distances from $s$ in $G'$, e.g., the one given in Theorem~\ref{thm:fd_distances} or in Theorem~\ref{thm:weighted_fd_distances}.
Suppose to be at the $i$-th iteration of the procedure, i.e., the super-source $s$ is connected to all vertices in $C = \{c_1,...,c_i\}$ in $G'$.
Note that such additional edges imply that $d_{G'}(s, v) = 1+ d_G(C, v)$, for every $v\in V$.
Let $\delta_{G'}(s,v)$ be the approximate distance between $s$ and $v$ maintained by $\mathcal{D}$, which guarantees that $d_{G'}(s,v) \le \delta_{G'}(s,v) \le (1+\epsilon)d_{G'}(s,v)$.
Let $v_{\max} \in \argmax_{v \in V}\, d_G(C, v)$ be one among the furthest vertices from $C$.
Let $c_{i+1}$ be the next center selected by the algorithm, i.e., $c_{i+1} \in \argmax_{v \in V} \delta_{G'}(s, v)$.
Therefore, it holds that
\begin{align*}
    1+ d_G(C, v_{\max}) 
    = d_{G'}(s, v_{\max}) 
    \le \delta_{G'}(s, v_{\max})
    &\le \delta_{G'}(s, c_{i+1}) 
    \\
    &\le (1+\epsilon)d_{G'}(C, c_{i+1})
    = (1+\epsilon)(1+d_G(C, c_{i+1})).
\end{align*}
Noting that $d_G(C, v_{\max}) \ge 1$ and since by assumption $\epsilon\in(0,1/2]$, the previous equation implies
\begin{multline*}
    d_G(C, c_{i+1}) 
    \ge \frac{1+d_G(C, v_{\max})}{1+\epsilon} - 1
    = \frac{d_G(C, v_{\max})-\epsilon}{1+\epsilon} 
    = \frac{1-\epsilon/d_G(C, v_{\max})}{1+\epsilon} d_G(C, v_{\max})
    \\
    \ge \frac{1-\epsilon}{1+\epsilon} d_G(C, v_{\max}) 
    \ge \frac{1}{1+4\epsilon} \max_{v \in V} d_G(C, v),
\end{multline*}
which concludes the approximation proof.

\bigskip
The update procedure requires that $\mathcal{D}$ is updated $2k+1$ times, with a worst-case time of $T(n,m, \epsilon)$ per update, and additionally look for the approximate furthest neighbor $k$ times, each requiring time $O(n)$, i.e., querying the approximate distance $\delta_G(s,v)$, $\forall v\in V$.
\end{proof}

In particular, for the fully dynamic data structure we use the state-of-the-art algorithm for unweighted graphs by \cite{BFN22}.
\begin{theorem}[\cite{BFN22}]\label{thm:fd_distances}
Given an unweighted undirected graph $G=(V,E)$ and a single source $s$, and $0 <\epsilon <1$, there is a deterministic fully dynamic data structure for maintaining $(1+\epsilon)$-distances from $s$ with worst-case update time of $O(n^{1.529} \epsilon^{-2})$ for the current matrix multiplication exponent $\omega$.
The algorithm has preprocessing time of $O(n^{\omega}\epsilon^{-2}\log \epsilon^{-1})$, where $\omega \le 2.373$.
\end{theorem}

For weighted graphs the state-of-the-art algorithm is slower and it is given by \cite{BrandN19}.
\begin{theorem}[\cite{BrandN19}]\label{thm:weighted_fd_distances}
Given a weighted and directed graph $G=(V,E)$, a single source $s \in V$, and a positive parameter $\epsilon <1$, there is a randomized fully dynamic algorithm working against an adaptive adversary that maintains $(1+\epsilon)$-distances from $s$ with worst-case update time of $O(n^{1.823} \epsilon^{-2})$ for the current matrix multiplication exponent $\omega$.
The algorithm has preprocessing time of $O(n^{\omega}\epsilon^{-2}\log \epsilon^{-1})$, where $\omega \le 2.373$.
\end{theorem}

A combination of Theorem~\ref{thm: fully dynamic black box} with Theorems~\ref{thm:fd_distances} and \ref{thm:weighted_fd_distances} gives the following result.
\fullydynamickcenter*{}

\appendix

\section*{Appendix}
\section{Dynamic $k$-Center Algorithms Queries}\label{apx:queries}

The dynamic algorithms for $k$-center we give in this paper can simply and efficiently answer queries of the following types:
\begin{enumerate}[topsep=0pt,itemsep=-1ex,partopsep=1ex,parsep=1ex]
    \item Return a set $S$ of at most $k$ centers and the corresponding radius $r$.
    \item Given a vertex $v \in V$, return the center $c \in S$ of the cluster $v$ belongs to.
    %\item Given a center $c \in S$, return the set of vertices which are within distance $r$ from $c$.
\end{enumerate}
The first type of queries simply returns the independent sets which have size $\le k$. 
The corresponding radius in the first query and the second type of queries can be answered using the dynamic shortest path data structures that we use. When maintaining distances from a super-source, the data structures let us keep the parent nodes along shortest paths which can be used for finding the closest source. The partially dynamic algorithms that we use are based on structures with $n^{o(1)}$ layers, such that we have a parent along the shortest path on each of these levels. Therefore without additional overhead we can keep track of the first parents along the path.

In our fully dynamic algorithm, we can keep track of the cluster center of each vertex by explicitly checking for each vertex whether its distance to the ``super-source'' changes with each iteration of the simulated Gonzalez's algorithm; the additional $ O(k n) $ overhead is already accounted for in our update time.

\section{Reduction from $2$-Approximate $k$-Center to $k$-Bounded Ruling Set: Omitted Proofs of Section~\ref{sec:reduction}}
\label{apx:reduction proofs}

This section is devoted to the omitted proofs of Section~\ref{sec:reduction}.

\kmisrtR*{}
\begin{proof}
    Consider an optimal solution $S^* = \{c^*_1, \dots, c^*_{k'}\}$ 
    of the $k$-center instance with $k' \leq k$, and let 
    $C_1^*, \dots, C_{k'}^*$ be the corresponding clusters, each of radius $R^*$. Let $M$ be an arbitrary 
    \tbrs in $G_r$, with $r \geq 2R^*$. We can assume w.l.o.g. that the set
    $M$ is ordered. The proof is by induction on the number of vertices of $M$. 
    The goal is to prove that for any $i \leq k'$, every vertex $v \in C_i^*$ is a neighbor
    of the $i_{th}$ vertex of $M$ in $G_r$. 
    This would imply then that $|M| \leq k' \leq k$, as otherwise
    there would be two vertices in $M$ which are neighbors in $G_r$, violating the
    fact that $M$ is a \tbrs in $G_r$ with $\beta \geq 1$
   
    As a base case, let $v_1$ be the first vertex of $M$, and
    assume w.l.o.g. that $v_1 \in C_1^*$. 
    Since all vertices in $C_1^*$ are within distance $R^*$ from $c_1^*$ in $G$, 
    by triangle inequality it holds that $d_G(v_1, u) \leq 2R^*$, for every vertex $u \in C_1^*$. 
    Hence as $r \geq 2R^*$, we have that every vertex $u \in C_1^*$ is a neighbor of $v_1$ in $G_r$. 

    Let $v_i \in M$ be the $i_\text{th}$ vertex of $M$.
    Since $M$ is a \tbrs in $G_r$ with $\beta \geq 1$, $v_i$ cannot be a neighbor of any other vertex that belongs to $M$.
    By inductive hypothesis, every vertex 
    $u \in C_1^* \cup \dots \cup C_{i-1}^*$ has a neighbor in $M$, and so $v_i$ cannot be part of
    $C_1^* \cup \dots \cup C_{i-1}^*$. As a result, we can assume w.l.o.g. that $v_i \in C_i^*$. By following the same approach as in the base case, we have that every vertex $u \in C_i^*$
    is a neighbor of $v_i$ in $G_r$, and so the claim follows.
\end{proof}

\noIStRGr*{}
\begin{proof}
    Suppose to the contrary that there is an IS $M$ in $G_r$ of size at least $k + 1$,
    and let $M'$ be an MIS on $G_r$ such that 
    $M \subseteq M'$. Then clearly it holds that $|M'| \geq k + 1$. Also since $\beta \geq 1$,
    we have that $M'$ is a \tbrs in $G_r$ of size at least $k + 1$. However based
    on Lemma~\ref{lem:k-MIS_r>=2R^*}, as $r \geq 2R^*$ and $\beta \geq 1$, the size
    of $M'$ must be at most $k$, and this yields a contradiction.
\end{proof}

\tbapproxGr*{}
\begin{proof}
Let $r'$ be the smallest $r \in \{(1 + \epsilon)^i \mid (1 + \epsilon)^i \leq nW, i \in \mathbb{N}\}$ such that a \kbtbrs algorithm running on $G_r$ returns a \tbrs $M_r$ of size at most $k$. 
Let $S = M_{r'}$ be the solution we return for the $k$-center instance.

Since $M_{r'}$ is a \tbrs in $G_{r'}$, then every vertex is within 
distance $\beta r'$ from its closest center in $G$. Thus, the returned solution $S$ has radius at most $\beta r'$.
We show now that $r'$ is at most $2(1+\epsilon)$ times larger than $R^*$.
Based on Observation~\ref{obs:noIStRGr}, for the fixed choice of $r = 2R^*$,
any \kbtbrs algorithm $\mathcal{A}$ running on the $r$-threshold graph $G_r$
always returns a \tbrs $M_r$ of size at most $k$.
By definition of $r'$, and since the possible values of $r$ are powers of $(1 + \epsilon)$,
we have that $r' \leq 2(1 + \epsilon)R^*$.
Therefore, the radius of the returned solution $S$ is at most $2\beta(1+\epsilon)R^*$.
\end{proof}

\section{Incremental $(1+\epsilon)$-SSSP}\label{app:incremental}
\newcommand{\dist}{d}

Most of the existing work on partially dynamic $(1+\epsilon)$-SSSP \cite{bernstein2009, HKN2014hopsets, chechik2018, LackiN22} is presented for the decremental setting, but while not explicitly written, the techniques extend to the incremental setting as well with the same running time. At a high-level these techniques first maintain a hopset (or similar objects like low-hop emulators) of size $\tilde{O}(m^{1+{o(1)}})$ and hopbound $h=n^{o(1)}$, and maintain an $h$-hop limited (ES) Even-Schiloach tree \cite{ES}. 

The $h$-hop limited ES tree algorithm \cite{bernstein2009} allows us to maintain $(1+\epsilon)$-approximate single-source shortest path up to $h$-hops (which finds the approximate shortest path using at most $h$ hops) in $\tilde{O}(mh)$ time. 

To use this subroutine several works utilize a hopset~\cite{HKN2014hopsets, chechik2018, LackiN22}. A $(\beta, \epsilon)$-hopset $H'$ for $G=(V,E)$ is a set of weighted edges such that for all $u,v \in V$, we have that $\dist_G(u,v) \leq \dist^{(\beta)}_{G\cup H'}(u,v)  \leq (1+ \epsilon)\dist_G(u,v)$, where $\dist^{(\beta)}_{G\cup H'}(u,v) $ refers to a shortest path that uses at most $\beta$ hops.

Much of the technical difficulty in the decremental setting is due to the fact that we have to insert hopset/emulator edges in a decremental data structure.
The existing decremental structures use an algorithm called monotone ES tree data structure \cite{HenzingerKN14} to handle this, however in the incremental setting a monotone ES tree is not needed. In an incremental setting, an update may require to decrease the weight of an edge or remove it from a hopset/emulator to keep the size small. Handling weight decreases is easy, as we can simply add a new edge with the smaller weight and keep the previous edges in place and this will only impact the number of edges by a logarithmic factor over the sequence of updates. The second issue of removing edges from a hopset/emulator, will also not impact the over all performance of the algorithm for the following reason: In these data structure we would only remove an auxiliary edge $e \in E(H)$ if the weight $w^{(t)}(e)$ (which corresponds to the length of a path in the original input graph $G$ at time $t$), is reduced by more than a constant factor so that it is within a factor of $(1- \epsilon')w^{(t-1)}(e)$. It is easy to see that in the incremental setting it will add a logarithmic factor in the size if we keep all of these edges and simply add new parallel edges and thus keep the data structures completely incremental.

%\subsection*{Acknowledgements}

\printbibliography[heading=bibintoc] % Make bibliography show up in table of contents

\end{document}